\documentclass[journal]{IEEEtran}
\IEEEoverridecommandlockouts

\usepackage[cmex10]{amsmath}
\allowdisplaybreaks[4]

\usepackage{enumitem}
\usepackage{graphicx}
\usepackage{color}
\usepackage{amsmath}
\usepackage{mathtools}
\usepackage{multicol}
\usepackage{multirow}
\usepackage[english]{babel}
\usepackage{blindtext}
\usepackage{algorithm}
\usepackage{algorithmic}
\usepackage{balance}
\usepackage{amsfonts}
\usepackage{bm}
\usepackage{stfloats}
\usepackage{subfig}
\usepackage{amsthm}
\usepackage{amssymb}
\usepackage{setspace}
\usepackage[nosort]{cite}
\usepackage{CJK}
\usepackage{cite}
\usepackage{caption}
\usepackage{comment}
\usepackage{array,multirow}
\usepackage{graphicx}

\newtheorem{theorem}{Theorem}

\newtheorem{remark}{Remark}
\newtheorem{proposition}{Proposition}
\theoremstyle{plain}

\usepackage[table]{xcolor}

\def\BibTeX{{\rm B\kern-.05em{\sc i\kern-.025em b}\kern-.08em
    T\kern-.1667em\lower.7ex\hbox{E}\kern-.125emX}}
\begin{document}
\captionsetup[figure]{labelformat={default},labelsep=period,name={Fig.}}

\title{Downlink Performance of Cell-Free Massive MIMO for LEO Satellite Mega-Constellation}

\author{Xiangyu Li, 
        and Bodong Shang,~\IEEEmembership{Member,~IEEE}
\thanks{This work was supported in part by Zhejiang Provincial Natural Science Foundation of China under Grant No. LQN25F010003, in part by the YongRiver Scientific and Technological Innovation Project No. 2023A-187-G, and in part by the State Key Laboratory of Integrated Services Networks.}
\thanks{Some intermediate results were presented in part at the International Conference on Ubiquitous Communication (Ucom) 2024 [DOI: 10.1109/Ucom62433.2024.10695881]. See \cite{li2024analytical} for details.}
\thanks{\textit{(Corresponding author: Bodong Shang)}}
\thanks{X. Li and B. Shang are with Zhejiang Key Laboratory of Industrial Intelligence and Digital Twin, Eastern Institute of Technology, Ningbo, China, and also with the State Key Laboratory of Integrated Services Networks, Xidian University, Xi’an, China (e-mails: xyli@eitech.edu.cn, bdshang@eitech.edu.cn).}}

\maketitle

\begin{abstract}
Low-earth orbit (LEO) satellite communication (SatCom) has emerged as a promising technology to improve wireless connectivity in global areas. Cell-free massive multiple-input multiple-output (CF-mMIMO), an architecture proposed for next-generation networks, has yet to be fully explored for LEO satellites. In this paper, we investigate the downlink performance of a CF-mMIMO LEO SatCom network, where multiple satellite access points (SAPs) simultaneously serve the corresponding ground user terminals (UTs). Using tools from stochastic geometry, we model the locations of SAPs and UTs on surfaces of concentric spheres using Poisson point processes (PPPs) and \textcolor{black}{present expressions on transmit and received signals, signal-to-interference-plus-noise ratio (SINR).}
Then, we derive the coverage probabilities in fading scenarios, considering significant system parameters such as the Nakagami fading parameter, the number of UTs, the number of SAPs, the orbital altitude, and the service range affected by the dome angle. Finally, the analytical model is verified by extensive Monte Carlo simulations. Simulation results indicate that stronger line-of-sight (LoS) effects and a more comprehensive service range of the UT result in a higher coverage probability, despite the presence of multi-user interference (MUI). Moreover, we found that there exist optimal numbers of UTs that maximize system capacity for different orbital altitudes and dome angles, providing valuable insights for system design.
\end{abstract}


\IEEEpeerreviewmaketitle

\section{Introduction}

\IEEEPARstart{T}{he} advancement of wireless communications and the explosion of high-data-rate demands have led to the expansion of network coverage, which is currently supported by terrestrial networks (TNs). As an essential component of non-terrestrial networks (NTNs), satellite communication (SatCom) is a critical enabler to realize the continuous and ubiquitous connectivity provision, especially in areas with inadequate wireless access \cite{li2021downlink}. In recent years, low-earth orbit (LEO) satellites, typically located at altitudes between $500$ km and $2,000$ km, have gained broad research interests due to their potential to provide worldwide Internet access with improved data rates \cite{you2020massive}. Compared with the medium-earth orbit (MEO) or geostationary-earth orbit (GEO) counterparts, LEO satellites are more advantageous due to relatively shorter signal propagation delay, decreased path-loss and power consumption, and lower production and launch costs \cite{you2022hybrid}.

In the LEO SatCom network, multi-beam transmission techniques have been widely used to obtain higher throughput and data rates for numerous user terminals (UTs) \cite{vazquez2016precoding}. One possible utilization of multi-beam transmission in SatCom is its integration with massive multiple-input multiple-output (MIMO) \cite{shang2024multi}, an enabling and widely adopted technology for the fifth-generation (5G) and beyond communication systems. In massive MIMO (mMIMO)-enhanced LEO SatCom networks, significantly higher data rates can be achieved than their traditional single-beam counterparts \cite{goto2018leo,femenias2022mobile}. 
However, similar to inter-cell interference of small cells in TNs with collocated mMIMO, inter-satellite interference (ISI) in NTNs exists due to the service area boundary and causes a severe performance downgrade for served UTs\cite{lu2014overview}. Therefore, the current satellite network, where a single satellite serves multiple UTs with potential overlap from neighboring satellite coverage areas, may be inadequate for UTs requiring high levels of coverage and data rates \cite{li2025advancing}. 
This directed our research toward a distributed mMIMO network that is theoretically well-suited for cell-free (CF) mMIMO LEO SatCom networks.

Stochastic geometry is a valuable mathematical tool for analyzing and evaluating network performance, such as coverage probability, desired signal, and interference statistics from a system-level perspective \cite{yang2020optimizing}. 
Specifically, coverage analysis research has become more tractable, such as the coverage probability in \cite{lu2015stochastic,fang2020millimeter,talgat2020stochastic} for TNs and in SatCom \cite{talgat2020stochastic,kim2024coverage} for NTNs.
However, the aforementioned works fail to account for the combined impacts of critical factors such as beamforming, fading parameters, number of satellites, and their orbital altitudes in CF-mMIMO LEO SatCom networks. 
Since coverage probability is functionally related to these factors, it can be expressed in closed-form equations.
Inspired by these observations, this paper develops an analytical model for the coverage and capacity of CF-mMIMO LEO SatCom networks using stochastic geometry.
These functional mapping relationships can significantly reduce computational complexity compared to system-level simulations.

\subsection{Related Works}
Many papers have studied the performance of LEO satellite networks.
The coverage probability for downlink transmissions was investigated in \cite{okati2020downlink}, where satellite locations were distributed as a binomial point process (BPP). However, crucial instruments such as the distance distributions of the BPP are more challenging to analyze compared to those of a Poisson point process (PPP) in unbounded space and therefore, PPP-based satellite modeling was adopted in \cite{park2022tractable}. 
Many other works also used PPP to model the constellation. For example, ultra-dense LEO satellite constellations were studied in \cite{deng2021ultra} to find the minimum number of satellites while satisfying the backhaul requirement of each UT. The uplink interference and performance of satellite mega-constellations were examined in \cite{jia2022uplink,jia2023analytic} under the impacts of intra-constellation interference. 
The above work, although similar to our studied network, assumed that each UT was connected to a single satellite, thus overlooking the possibility of satellite cooperation.

The CF-mMIMO system, where each UT is coherently served in the same time-frequency resource block (RB) by a large number of geographically-located access points (APs) \cite{ngo2017cell}, has gained extensive research attention recently. 
Although there has been extensive research on CF-mMIMO in TNs such as \cite{zhao2023energy,sun2023uplink}, relatively few studies have explored its application in LEO SatCom. 
An LEO satellite was initially considered an ``add-on'' to improve terrestrial CF networks in \cite{liu2020cell,li2021satellite}. However, they fail to introduce the CF-mMIMO architecture to large-scale LEO satellite networks.
The authors in \cite{abdelsadek2021future,abdelsadek2022distributed} integrated CF-mMIMO into LEO satellite networks and developed an optimization framework to improve spectral efficiency. 
The benefits of CF-mMIMO in LEO satellite networks were further evaluated in \cite{abdelsadek2023broadband} for broadband connectivity with multi-antenna satellites and handheld devices. 
In addition, the CF-mMIMO SatCom network was also studied from physical layer orthogonal time-frequency-space (OTFS) \cite{yang2023distributed}, uplink transmissions with channel state information (CSI) uncertainties \cite{omid2023capacity,omid2023space,omid2024tackling}, and dynamic clustering of satellite access points (SAPs) \cite{humadi2024distributed}.
The authors in \cite{kim2024cell} investigated CF massive NTNs for optimization of spectral efficiency and coverage probability, while overlooking global capacity analysis.
Note that the above works focused on algorithm designs and assumed that the served UTs were confined in a small region on the earth, neglecting other coverage areas where inter-user interference could arise, as well as downlink interference from non-cooperative satellites.
Our previous work \cite{li2024analytical} examined a basic model for coordinated multi-satellite joint transmissions. \textcolor{black}{However, its reliance on Rayleigh fading made it impractical for channels with line-of-sight (LoS) components.}
\textcolor{black}{Although satellite cooperation was considered by the authors in \cite{shang2023coverage,richter2020downlink}, ISI persisted, their cooperative satellite number remains limited and cannot form CF-mMIMO for LEO satellite networks.}
In addition, its tractability decreased as the number of satellites increased.

\subsection{Contributions and Paper Organization}
To thoroughly investigate the downlink performance of CF-mMIMO LEO SatCom networks from a system-level perspective, we incorporate essential factors of satellite networks for analysis using tools from stochastic geometry. 
The PPP model is adopted to model LEO satellite mega-constellation, providing analytical results with broad applicability and high tractability.
The main contributions are summarized as follows:
\begin{itemize} 
\item \textit{CF-mMIMO LEO SatCom Network Design:} 
We establish a CF-mMIMO system in the LEO satellite network using stochastic geometry. 
The locations of SAPs and UTs are modeled on two concentric sphere surfaces, one with the SAPs' orbital radius and the other with the earth's radius, respectively, based on two independent homogeneous PPPs. Simulation results demonstrate that the PPP constellation model can fit perfectly with practical constellations. Considering the shape of the earth and the service range, each UT is assigned a group of SAPs for service delivery according to its dome angle.

\item \textit{Modeling and Analysis:}
We introduce the related distance distribution for the service links of a typical UT and characterize the average number of UTs served by each SAP.
Then, we provide expressions for the distribution of desired signal strength (DSS), the average ISI, and the average multi-user interference (MUI).
The coverage probability of a typical user is then presented accordingly, with both large-scale path-loss and small-scale Nakagami-$m$ fading. 

\item \textit{Network Design Insights:} 
We quantitatively analyze the performance of the CF-mMIMO LEO SatCom network concerning key network parameters, including the fading parameter, the service range of CF-mMIMO LEO SatCom network, the total number of SAPs, the number of UTs, and SAPs' orbital altitudes.
As the dome angle of the typical UT increases, the gain in desired signal power outweighs the increase in received MUI, leading to improved coverage probability, albeit with increased signaling overhead.
In addition, an increase in both the orbital altitude and the number of SAPs contributes positively to enhanced coverage performance. While system capacity increases as the number of UTs grows within a certain range, the per-user capacity continues to drop, exhibiting a tradeoff between maximizing system capacity and ensuring minimum throughput for each UT.
\end{itemize}

The rest of this paper is organized as follows. 
Section \uppercase\expandafter{\romannumeral2} describes the system model for the studied CF-mMIMO SatCom network, including spatial distributions, channel models, and downlink transmissions. 
In Section \uppercase\expandafter{\romannumeral3}, we present the statistical properties of the network and derive analytical expressions for the coverage probabilities. 
Simulations and numerical results are provided in Section \uppercase\expandafter{\romannumeral4}.
Finally, Section \uppercase\expandafter{\romannumeral5} concludes the paper.

\textit{Notations:}
Bold lowercase letters denote column vectors. 
The superscripts $\left ( \cdot  \right )^T$, $\left(\cdot\right)^{*}$, and $\left(\cdot\right)^H$ are used to represent conjugate, transpose, and conjugate-transpose, respectively.
The real-part operator, expectation operator, the statistical probability, and the Euclidean norm are represented as $\mathrm{Re}[\cdot]$, $\mathbb{E}\left\{\cdot\right\}$, 
$\mathbb{P}\left\{\cdot\right\}$ and $\|\cdot\|$, respectively.
$y \sim \mathcal{CN}(\mu,\sigma^2)$ denotes a complex circularly symmetric Gaussian random variable with mean $\mu$ and variance $\sigma^2$, while $\bm{y} \sim \mathcal{CN}(\mu,\mathbf{R})$ denotes that with mean $\mu$ and correlation matrix $\mathbf{R}$.

\vspace{-3mm}
\section{System Model}
In this section, we consider a downlink CF-mMIMO LEO SatCom network over the Nakagami-$m$ fading channel, where many SAPs simultaneously serve the UTs on the earth's surface. For the sake of analysis, each SAP and each UT are equipped with a single antenna. 
Following the basic principle of CF-mMIMO, SAPs are connected to a central server (CS) through high-speed optical inter-satellite links (OISLs)\footnote{High-speed OISLs are well-suited for this scenario due to their high throughput and low latency characteristics \cite{shang2025channel}, making them ideal for backhaul communications.}, known as backhaul links, to exchange necessary control signals for cooperative transmissions of SAPs. 

The CS can be deployed on a central satellite, e.g., a GEO satellite with sufficient power and computing capabilities, following a similar mechanism as in \cite{abdelsadek2021future,abdelsadek2022distributed,abdelsadek2023broadband,humadi2024distributed,shang2023coverage}.
Regarding the scalability of this network, the OISL infrastructure can support large numbers of SAPs, especially when combined with future hierarchical satellite network topologies and advanced scheduling algorithms.

\captionsetup{font={scriptsize}}
\begin{figure*}[t]
\begin{center}
\centering
\vspace{-2mm}
\setlength{\abovecaptionskip}{+0.2cm}
\setlength{\belowcaptionskip}{+0.0cm}
\centering
  \includegraphics[width=6.65in, height=2.0in]{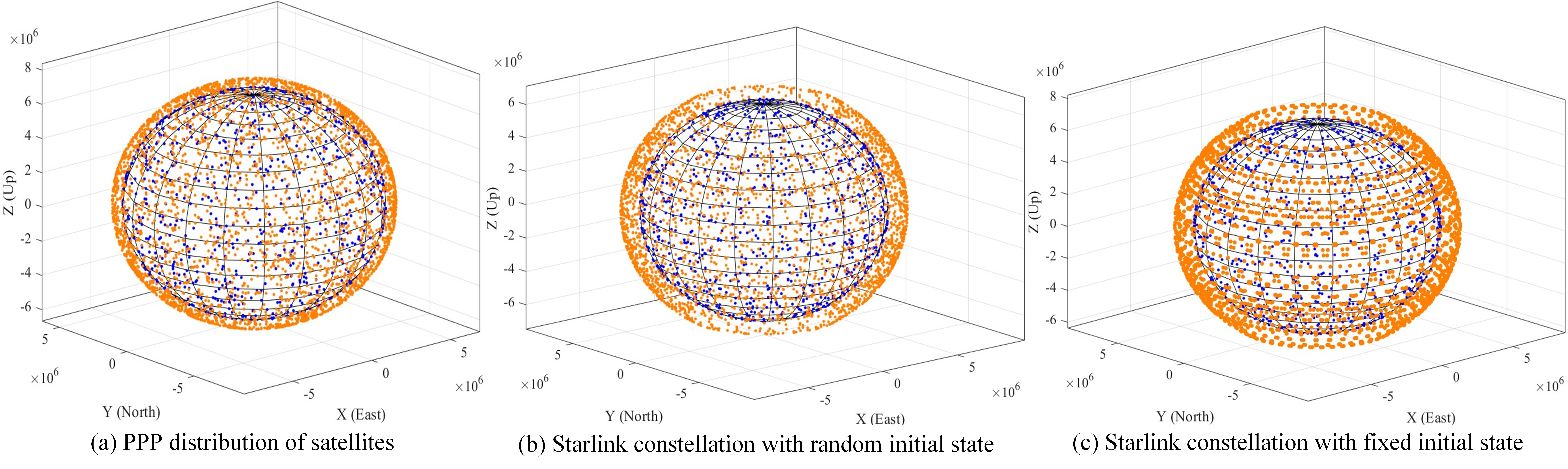}
\renewcommand\figurename{Fig.}
\caption{\scriptsize An illustration of different satellite constellations.}
\label{fig:constellations}
\end{center}
\vspace{-6mm}
\end{figure*}

\vspace{-4mm}
\subsection{Spatial Distribution Model}
Denote the radius of the Earth as $R_\text{E}$, and that of a sphere $S$ where the SAPs are located as $R_\text{S}$. 
An illustration of PPP-distributed satellites and PPP-distributed UTs are shown in Fig. \ref{fig:constellations}(a).
The distributions of UTs and SAPs are independent and each follows a homogeneous spherical Poisson point process (SPPP).
In a practical constellation such as Starlink, the distribution of satellites is correlated.
For comparison, we simulate the Starlink constellations with random and fixed initial states, respectively.
Note that in a random initial state shown in Fig. \ref{fig:constellations}(b), the initial position and orbital parameters of each satellite are randomly distributed within a certain range, which can simulate the randomness of constellations at the time of deployment. 
The satellite distribution in Fig. \ref{fig:constellations}(b) closely resembles that observed in practical constellations.
In a fixed initial state shown in Fig. \ref{fig:constellations}(c), the initial position and orbital parameters of all satellites are set in advance and strictly fixed to represent the ideal precise deployment situation. The satellites follow strict geometric symmetries in their orbital planes and phases in order to minimize relative deviations between satellites. 
Intuitively, the PPP model can perfectly simulate a practical constellation by setting a random initial state, which will be verified in Section \uppercase\expandafter{\romannumeral5}.

On this basis, we assume that there are $L$ SAPs in the constellation at the same orbital altitude $H_\text{S}=R_\text{S}-R_\text{E}$, and there are a total of $K$ UTs on the earth surface that need to be served by the CF-mMIMO SatCom network. Taking the mean of the Poisson random variable, the PPP density of UTs $\lambda_\text{U}$ and that of SAPs $\lambda_\text{S}$ is approximated by $\lambda_\text{U} = \frac{K}{4\pi {R_\text{E}}^2}$ and $\lambda_\text{S} = \frac{L}{4\pi {R_\text{S}}^2}$, respectively.

\vspace{-3mm}
\subsection{Path-Loss and Channel Model}
For satellite-terrestrial signal transmissions, both large-scale path-loss attenuation and small-scale fading are taken into account. Specifically, a valid assumption is that the fading channel between every two links is assumed to be uncorrelated, since the SAPs are geographically distributed on the sphere $S$. 
The channel coefficient between the $l$-th SAP and the $k$-th UT is denoted as 
\begin{equation}
    g_{lk}={\beta_{lk}}^{1/2} h_{lk},
\end{equation}
where $\beta_{lk}$ denotes the large-scale fading and $h_{lk}$ is the small-scale fading.
For large-scale fading, $\beta_{lk}$ can be represented using a distance-dependent model, i.e. $\beta_{lk}=\beta_0 {d_{lk}}^{-\alpha}$, where $\beta_0$ is the path-loss at a reference distance, $d_{lk}$ is the distance between the $l$-th SAP and the $k$-th UT, and $\alpha$ is the path-loss exponent\footnote{For analytical tractability and to focus on the performance of the proposed framework, we assume a common path-loss exponent $\alpha$ for all satellite-user links. This simplification is commonly adopted in 3GPP standards \cite{3gpp.38.811,3gpp.38.821} and also in prior satellite communication works, especially mega, ultra-dense constellations \cite{deng2021ultra,jia2022uplink,jia2023analytic}. In real-world scenarios, the exponent may vary with different environmental conditions; incorporating heterogeneous path-loss exponents is a valuable direction for future research.}. 
Since antenna gain varies according to the off-axis angle \cite{3gpp.38.811}, the transmit antenna gain of the desired signal and that of the interference signal should be modeled differently \cite{park2022tractable}. 
Denote the transmit main-lobe or side-lobe antenna gain of the SAP as $G_\text{t}^{\text{ml/sl}}$. Then, the total effective antenna gain with main-lobe or side-lobe components can be written as
\begin{equation}
    G_{\text{ml/sl}} = G_\text{t}^{\text{ml/sl}} G_\text{r} \left( \frac{c}{4\pi f_\text{c}} \right)^2,
\end{equation}
where $G_\text{r}$ is the receive antenna gain of the UT, 
$c$ is the speed of light, and $f_\text{c}$ is the carrier frequency.
The small-scale fading is modeled via the Nakagami-$m$ distribution, whose probability density function (PDF) is given by
\begin{equation}
    f_{|h_{lk}|}(x) = \frac{2m^m}{\Omega^m\Gamma(m)}x^{2m-1}e^{-mx^2},
    \label{Formula:Nakagami-m}
\end{equation}
with $x\in [0,\infty]$, $m\in [0.5,\infty]$ and $\Gamma(m)=(m-1)!$ for integer $m>0$. The Nakagami-$m$ model is selected because of its adaptation to various small-scale fading conditions, e.g., Rayleigh channel when $m=1$ and Rician-$\mathcal{K}$ channel when $m=\frac{(\mathcal{K}+1)^2}{2\mathcal{K}+1}$  \cite{park2022tractable}. 
We also have $\mathbb{E}\left\{ {{\left| {{h}_{lk}} \right|}^{2}} \right\} = \Omega = 1$.
Note that Shadowed-Rician (SR) fading can also be approximated by Nakagami-$m$ fading \cite{jia2023analytic}.
Changing this parameter $m$ makes it suitable for describing various channels and modeling signal propagation conditions from severe to moderate.

\subsection{Uplink Training}

In a CF-mMIMO system, propagation channels are considered to be piecewise constant during each coherence time period and frequency coherence interval. 
In particular, since only downlink transmission is considered herein, the proposed model is applicable for both time division duplexing (TDD) and frequency division duplexing (FDD) with channel estimation \cite{ngo2017cell}.
Denote a coherence block of length $\tau_\text{c}$, where the durations of $\tau_\text{p}$ and $\tau_\text{d}=\tau_\text{c}-\tau_\text{p}$ are dedicated to uplink training and downlink transmission, respectively. In the uplink training stage, let $\sqrt{\rho_\text{p}}\bm{\varphi}_{k} \in \mathbb{C}^{\tau_\text{p} \times 1}$ be the pilot sequence transmitted by the $k$-th UT, where $||\bm{\varphi}_{k}||^2 = \tau_\text{p}$, and $\rho_\text{p}$ represents the transmit power of the pilot symbols. 
Denote the set of all served UTs of the $l$-th SAP as $\Phi_{l}^\text{U}$, the received signal at the $l$-th SAP is
\begin{equation}
    \bm{y}_{\text{p},l} 
    = \sqrt{\tau_\text{p} \rho_\text{p}}\sum_{k\in \Phi_{l}^\text{U}}{\beta_{lk}}^{1/2}h_{lk}\bm{\varphi}_{k}^{H} + \bm{w}_{\text{p},l},
    \label{Formula:uplink_training}
\end{equation}
where $\bm{w}_{\text{p},l} \in \mathbb{C}^{1 \times \tau_p}$ is the receiver additive white Gaussian noise (AWGN) vector whose elements are i.i.d. as $\mathcal{CN}(0,\sigma^2)$.
The post-processing signal is obtained by projecting $\bm{y}_{\text{p},l}$ onto $\bm{\varphi}_{k}$, which is represented as
\begin{equation}
    \begin{aligned}
        \bar{y}_{\text{p},l,k} 
        &= \bm{y}_{\text{p},l}\bm{\varphi}_{k} \\
        &= \sqrt{\tau_\text{p} \rho_\text{p}}\sum_{k'\in \Phi_{l}^\text{U}}\beta_{lk'}^{1/2}h_{lk'}\bm{\varphi}_{k'}^{H}\bm{\varphi}_{k} + \bm{w}_{\text{p},l}\bm{\varphi}_k \\  
        &= \begin{aligned}[t] 
        &\sqrt{\tau_\text{p} \rho_\text{p}} \beta_{lk}^{1/2}h_{lk} \\
        &+ \sqrt{\tau_\text{p} \rho_\text{p}} \sum_{ \substack{ k'\in\Phi_{l}^\text{U} \\ k'\neq k} }\beta_{lk'}^{1/2}h_{lk'}\bm{\varphi}_{k'}^{H}\bm{\varphi}_{k} + \bm{w}_{\text{p},l}\bm{\varphi}_k. 
        \end{aligned}
    \end{aligned}
    \label{Formula:projection}
\end{equation}
The estimated channel $\hat{g}_{lk}$ using linear minimum mean-square-error (LMMSE) algorithm can then be written as
\begin{equation}
    \begin{aligned}
        \hat{g}_{lk} 
        &= \frac{\mathbb{E}\{\bar{y}_{\text{p},l,k}{g}_{lk}\}}{\mathbb{E} \{|\bar{y}_{\text{p},l,k}|^2\}} \bar{y}_{\text{p},l,k} \\
        &= \begin{aligned}[t]
            &\frac{\sqrt{\tau_\text{p}\rho_\text{p} }\beta_{lk}\mathbb{E}\{h_{lk}h_{lk}^{H}\}}{\tau_\text{p}\rho_\text{p}  \sum_{{k'}\in\Phi_{l}^\text{U}}\beta_{lk'}|h_{lk'}|^2 |\bm{\varphi}_{k'}^{H}\bm{\varphi}_{k}|^2 + \sigma^2} \\
            &\times \left( \sqrt{\tau_\text{p} \rho_\text{p} }\sum_{k'\in \Phi_{l}^\text{U}}\beta_{lk'}^{1/2}h_{lk'}\bm{\varphi}_{k'}^{H}\bm{\varphi}_{k} + \bm{w}_{\text{p},l}\bm{\varphi}_k \right)
        \end{aligned} \\
        &= \begin{aligned}[t] 
        &\beta_{lk} \mathbb{E} \{h_{lk}h_{lk}^{H}\} \\ &\times \begin{aligned}[t]
            & \left( \frac{\sum_{k'\in \Phi_{l}^\text{U}} \beta_{lk'}^{1/2}h_{lk'}\bm{\varphi}_{k'}^{H}\bm{\varphi}_{k}}{\sum_{{k'}\in\Phi_{l}^\text{U}}\beta_{lk'}|h_{lk'}|^2 |\bm{\varphi}_{k'}^{H}\bm{\varphi}_{k}|^2 + \frac{\sigma^2}{\tau_\text{p}\rho_\text{p} }} \right.\\ 
            &\left. + \frac{\sqrt{\tau_\text{p}\rho_\text{p} }\bm{w}_{\text{p},l}\bm{\varphi}_{k}}{\tau_\text{p}\rho_\text{p} \sum_{{k'}\in\Phi_{l}^\text{U}}\beta_{lk'}|h_{lk'}|^2 |\bm{\varphi}_{k'}^{H}\bm{\varphi}_{k}|^2 + \sigma^2} \right).
            \end{aligned}
        \end{aligned}
    \end{aligned}
    \label{Formula:estimated_channel}
\end{equation}

\subsection{Downlink Transmission}
During the downlink transmission stage, SAPs treat the channel estimates as the actual channels to perform conjugate beamforming, allowing them to simultaneously send data symbols to UTs within their respective service areas. 
Focusing on the studied network, we employ equal power allocation at the transmitting SAP for each served UT by using conjugate beamforming. 
Specifically, denoting the total transmit power of each SAP as $\rho_\text{d}$, and the number of UTs within the $l$-th SAP's service area as $\left|\Phi_{l}^\text{U}\right|$, we have $\rho_{lk}^\text{d}=\frac{\rho_\text{d}}{\left|\Phi_{l}^\text{U}\right|}$ represent the transmit power from $l$-th SAP to $k$-th UT.
Then, the transmitted symbols from the $l$-th SAP can be expressed as
\begin{equation}
    \begin{aligned}
    x_{l} 
    &= \sqrt{G_{\text{ml}}} \sum_{k\in \Phi_{l}^\text{U}} \sqrt{\rho_{lk}^\text{d}} \frac{\hat{g}_{lk}^{*}}{|g_{lk}|} q_k \\
    &=\sqrt{G_{\text{ml}}}\sum_{k\in \Phi_{l}^\text{U}} \sqrt{\rho_{lk}^\text{d}} \frac{{\beta_{lk}}^{1/2} \hat{h}_{lk}^{*}}{|{\beta_{lk}}^{1/2} h_{lk}|} q_k \\
    &=\sqrt{G_{\text{ml}}}\sum_{k\in \Phi_{l}^\text{U}} \sqrt{\rho_{lk}^\text{d}} \frac{\hat{h}_{lk}^{*}}{|h_{lk}|} q_k,      
    \end{aligned}   
    \label{Formula:transmitted_signal_from_SAP}
\end{equation}
where $q_k$ is the dummy or data symbol intended for the $k$-th UT with 
$\mathbb{E}\{ |q_k|^2 \} = 1$, and $\hat{h}_{lk}$ is the estimated channel of $h_{lk}$ using LMMSE. 
The signals received by the $k$-th UT is 
\begin{equation}
    \begin{aligned}
        r_k^\text{d} 
        &= \sum\limits_{l\in \Phi _{k}^{\text{Ser}}}{{{\beta }_{lk}}^{1/2}{{h}_{lk}}{{x}_{l}}} + \sum\limits_{l\in \Phi _{k}^{\text{Int}}}{\sqrt{{{\rho }_{\text{d}}}G_{\text{sl}}}{{\beta }_{lk}}^{1/2}{{h}_{lk}}{{q}_{k}}} + w_{\text{d},k} \\
        &= \begin{aligned}[t]
                &\underbrace{ \sum_{l\in\Phi_{k}^\text{Ser}} \sqrt{\rho_{lk}^\text{d} G_{\text{ml}}} {\beta_{lk}}^{1/2} h_{lk} \frac{\hat{h}_{lk}^{*}}{|h_{lk}|} q_k }_{\mathrm{DS}_k} \\
                &+\underbrace{ \sum_{l\in \Phi_{k}^\text{Ser}} \sum_{\substack{ k'\in\Phi_{l}^\text{U} \\ k'\neq k}} \sqrt{\rho_{lk}^\text{d} G_{\text{ml}}} {\beta_{lk}}^{1/2} h_{lk} \frac{\hat{h}_{lk'}^{*}}{|h_{lk'}|} q_{k'} }_{\mathrm{MUI}_k} \\
                &+\underbrace{ \sum\limits_{l\in \Phi _{k}^{\text{Int}}}{\sqrt{{{\rho }_{\text{d}}}G_{\text{sl}}}{{\beta }_{lk}}^{\frac{1}{2}}{{h}_{lk}}{{q}_{k}}} }_{\mathrm{ISI}_k}
                + w_{\text{d},k},
            \end{aligned}
    \end{aligned}
    \label{Formula:received_signal_at_UT}
\end{equation}
where $\Phi_{k}^\text{Ser}$ and $\Phi_{k}^\text{Int}$ are the set of service SAPs and interference SAPs, $\Phi_{l}^\text{U}$ is the set of served UTs within $l$-th SAP's service area, and $w_{\text{d},k} \sim \mathcal{CN}(0,\sigma^2)$ is the downlink AWGN. Moreover, $\mathrm{DS}_k$, $\mathrm{MUI}_k$, and $\mathrm{ISI}_k$ represent the desired signal (DS), the MUI, and the ISI, received at the $k$-th UT, respectively.

\begin{remark}
    Note that for a total of $T$ pilot sequences, $\sqrt{\rho_\text{p}}\bm{\varphi}_1, \sqrt{\rho_\text{p}}\bm{\varphi}_2,...,\sqrt{\rho_\text{p}}\bm{\varphi}_T$ are mutually orthogonal. In general, the number of pilots used for channel estimation can be limited and smaller than that of UTs, i.e., $T<K$, so that mutually non-orthogonal pilots have to be used in the network. 
    Each SAP hopes to receive different orthogonal pilots from UTs within its service area so that it can differentiate UTs for better communication accuracy and quality; 
    however, repeated pilots must be used when the number of UTs served is greater than that of the pilots. 
    Due to the existence of the second term in (\ref{Formula:projection}), which indicates that pilots transmitted from other UTs, the estimated channel $\hat{g}_{lk}$ in (\ref{Formula:estimated_channel}) will degrade the estimation performance, which is known as pilot contamination \cite{lu2014overview}. 
    \textcolor{black}{However, numerical observations in the following Fig. \ref{fig:Fig_pilotnbr} show that the estimation error becomes negligible with sufficient pilot resources, making perfect CSI at the transmitter (CSIT) a reasonable approximation for analytical tractability.}
\end{remark}

\begin{remark}
    In this paper, considering the global service coverage of CF-mMIMO SatCom networks and the transmission structure defined in 5G new radio (NR) standards \cite{3gpp.38.811}, we assume a larger number of pilot sequences and a longer coherence time interval than those typically used in terrestrial CF-mMIMO systems. This assumption is justified on the basis of the following considerations.
    \textcolor{black}{
    First, the coherence time in LEO satellite networks is predominantly influenced by the Doppler spread (i.e., maximum Doppler shift), which reflects the range of frequency shifts induced by satellite motion across different angular directions. 
    The Doppler spread can be estimated by $f_D = \frac{v_s}{\lambda} \text{cos}\theta$, where $v_s$ is the satellite-user relative velocity, $\lambda$ is the signal wavelength, $\theta=60^{\circ}$ is a typical elevation angle. 
    For example, using $v_s \approx 7.5$ km/s and a carrier frequency of $f_c = 2$ GHz (corresponding to $\lambda = \frac{c}{f_c} = 0.15$ m), the Doppler spread is approximately $f_D = \frac{v_s}{\lambda} \text{cos}\theta = 25$ kHz. Then, the coherence time is calculated as $T_c \approx \frac{1}{f_D} = 40$ \text{$\mu$s}. 
    However, using coarse Doppler frequency offset schemes, such as those in \cite{wang2020iterative}, the center of the Doppler power spectrum can be shifted close to zero. Specifically, with a root-mean-square-error (RMSE) of approximately $10^{-4}$, the coherence time can be offset to as long as $T_c \approx \frac{1}{f_D} = \frac{1}{25\text{kHz} \cdot 10^{4}} = 400$ ms. This is longer than an example coherence time of $100$ ms and sufficient to perform uplink training and downlink transmission operations.} 
    Second, based on the 5G NR-based transmission framework, an RB consists of $12$ subcarriers in the frequency domain and $14$ OFDM symbols in the time domain, with a time slot duration of $1$ ms for $15$ kHz subcarrier spacing. Then, the coherence time of $100$ ms will correspond to approximately $1400$ OFDM symbols.
    Thus, our choices of pilot training length $\tau_p = 200$ symbols and coherence block length $\tau_c = 500$ symbols are within the coherence time constraints imposed by Doppler characteristics. 
\end{remark}

In LEO satellite communications, Doppler shifts arise due to high-speed satellite motion. However, the LOS-dominant nature of satellite-user channels enables effective Doppler shift estimation and compensation using satellite ephemeris and terminal tracking. While the channel exhibits LOS characteristics, small-scale fading, e.g., Nakagami-$m$ fading in this paper, is included in the analysis to statistically characterize residual multi-path effects. This modeling approach is consistent with standard assumptions in \cite{park2022tractable,deng2021ultra,jia2022uplink,jia2023analytic} and 3GPP guidelines \cite{3gpp.38.811}.

\subsection{Operation Mechanism}
In this subsection, we provide technical explanations for the possible considerations of the proposed model from an implementation perspective.

\subsubsection{Delay and Phase Misalignment}
The distances from multiple satellites to a user are different. When they cooperatively transmit signals to the user, there will be problems of delay and phase misalignment for the received signals. 
However, the synchronization time can be accurate to the nanosecond level.
With the distances calculated by the locations of both satellites and users, and the estimated electromagnetic propagation velocity, such differences can be compensated for delays and phase misalignment \cite{yang2018progress,chunhao2013time}. 
In addition, online satellite clock synchronization using global positioning system (GPS) signals or time synchronization protocols can also help mitigate such problems \cite{zhou2011orbit}. Distributed synchronization protocols can help ensure that delay and phase alignment is maintained across multiple satellites with minimal error \cite{wang2022distributed}.

\subsubsection{Payload Type and Functional Split}
In this model, regenerative payloads are considered for satellites since functionalities such as on-board signal processing are required. This will reduce dependency on ground infrastructure, but may introduce extra latency which can be managed through synchronization protocols.
In terms of the functional split, the lower NR layers are performed on the service satellite to conduct signal processing and enable more direct communications with the UTs. The higher NR layers are performed by the CS, which has been mentioned earlier, for scheduling and routing of service satellites, etc. The tradeoffs from allocation of these functions require future investigation.

\subsubsection{Power Allocation}
Uniform power allocation is considered for its simplicity and to provide an initial evaluation of system performance. 
The corresponding results in our paper can provide a performance lower bound.
Other adaptive power allocation algorithms \cite{guidotti2024federated}, with different purposes of energy efficiency, per-user data rate, or system capacity, would better exploit space-ground channel variations and improve the system performance of LEO satellite mega-constellations.

\section{Performance Analysis for Downlink Transmissions}
In this section, we first present a series of statistical properties as ground truth expressions.
Subsequently, the complementary cumulative distribution function (CCDF) of the DSS, the average MUI, and the average ISI are provided. Finally, approximate expressions for the coverage probability are derived.

\captionsetup{font={scriptsize}}
\begin{figure}[t]
\begin{center}
\setlength{\abovecaptionskip}{+0.2cm}
\setlength{\belowcaptionskip}{-0.2cm}
\centering
  \includegraphics[width=3.2in, height=1.9in]{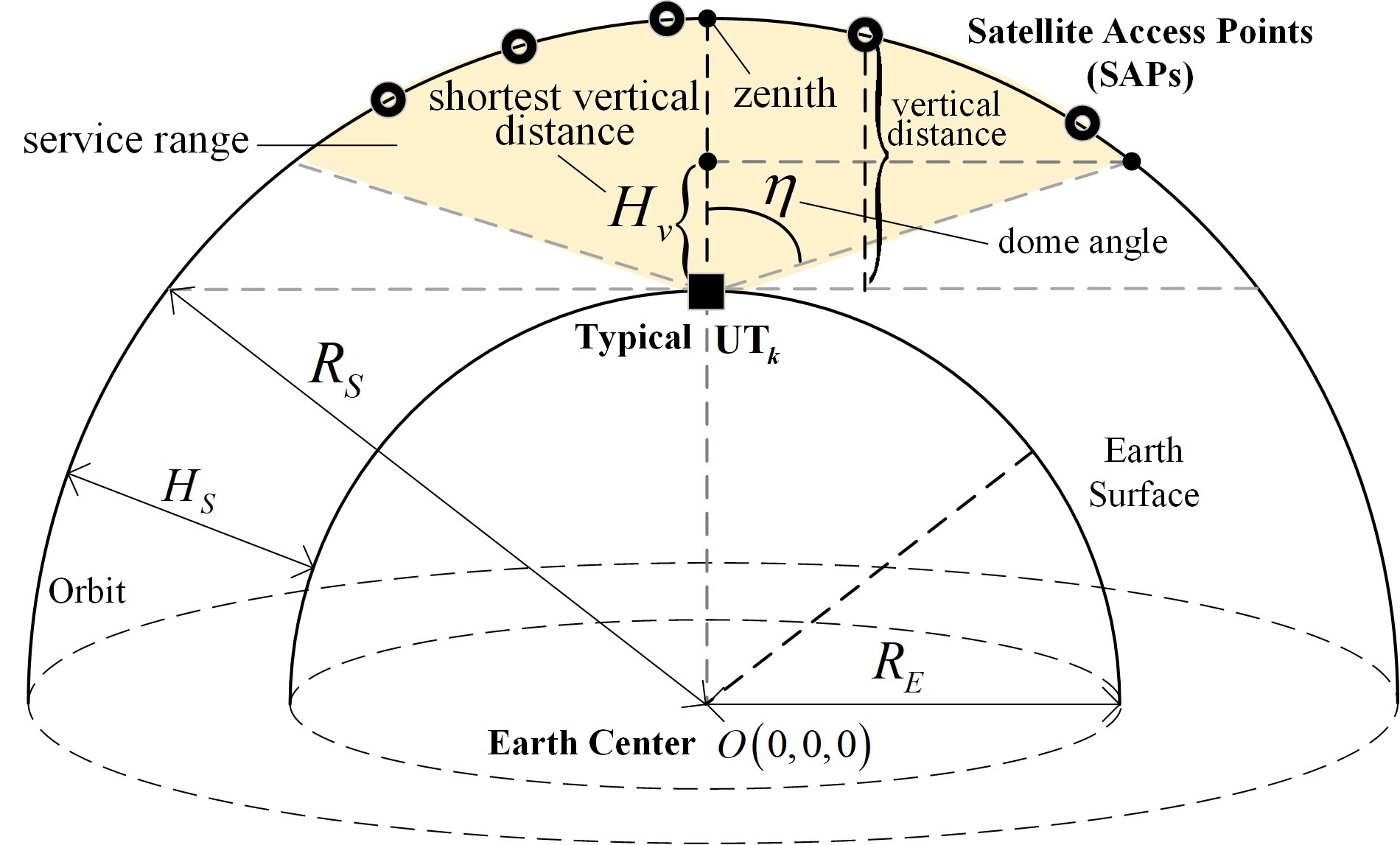}
\renewcommand\figurename{Fig.}
\caption{\scriptsize A sketch of stochastic geometry modeling in the system.}
\label{fig:Fig_sphere}
\end{center}
\end{figure}

\subsection{Statistical Properties}
Related parameters in the stochastic geometry sketch of UTs and SAPs are labeled in Fig. \ref{fig:Fig_sphere}. 
According to the 3GPP Technical Standard \cite{3gpp.38.821}, each ground UT has a maximum service range for satellites.
Thus, the typical UT has a dome angle $\eta$ from its position and corresponds to the shortest vertical distance $H_v$ from the UT. 
We denote the area marked in light yellow as a ``service range'' corresponding to $\eta$. If an SAP is within the service range, that is, if its vertical distance from the UT is longer than $H_v$, it is considered a service SAP. The CS coordinates a cluster of SAPs within the service range to serve the UT coherently.

Let the distance from the typical UT to any one of the service SAPs be $D$. 
According to \cite[Lemma 1]{li2024analytical}, the cumulative distribution function (CDF) of $D$ is given in (\ref{Formula:CDF_d}) \textcolor{black}{at the top of the page},
\begin{figure*}[ht]
\setlength{\abovecaptionskip}{-1.5cm}
\setlength{\belowcaptionskip}{-0.5cm}
\normalsize
\begin{equation}
\begin{aligned}
  F_{D}(d) = \begin{cases} 0, & 0<d<(R_\text{S}-R_\text{E}) \\ \frac{d^2-(R_\text{S}-R_\text{E})^2}{2 R_\text{E}(R_\text{S}-R_\text{E} \sin^2\eta - \sqrt{ {R_\text{S}}^2 -{R_\text{E}}^2\sin^2\eta } \cos\eta)}, & (R_\text{S}-R_\text{E})\leq d \leq \sqrt{ {R_\text{S}}^2 -{R_\text{E}}^2\sin^2\eta }-R_\text{E}\cos\eta \\1, & d>\sqrt{ {R_\text{S}}^2 -{R_\text{E}}^2\sin^2\eta }-R_\text{E}\cos\eta \end{cases}
    \label{Formula:CDF_d}
\end{aligned}
\end{equation}
\hrulefill
\end{figure*}
where $\eta \in [0,\frac{\pi}{2}]$, and the corresponding PDF is given by
\begin{equation}
\begin{aligned}
    & f_D(d) 
     = \frac{d}{R_\text{E}(R_\text{S} - R_\text{E} \sin^2\eta - \sqrt{ {R_\text{S}}^2 -{R_\text{E}}^2\sin^2\eta } \cos\eta)},
    \label{Formula:PDF_d}
\end{aligned}
\end{equation}
for $r_{\text{S,min}}\leq d \leq r_{\text{S,max}}$, while $f_D(d) = 0$, otherwise. Note that $r_{\text{S,min}} = {R_\text{S}} - {R_\text{E}}$ and $r_{\text{S,max}} = \sqrt {{R_\text{S}}^2 - {R_\text{E}}^2{{\sin }^2}\eta }  - {R_\text{E}}\cos \eta$ are the minimum and maximum distances from the typical UT to a service SAP, respectively, and ${{r}_{\max}}=\sqrt{R_\text{S}^{2}-R_\text{E}^{2}}$ is the maximum distance from the typical UT to any SAP \cite{park2022tractable}.

Next, the shortest vertical distance from service SAPs to the typical UT is given by \cite[Lemma 2]{li2024analytical}
\begin{equation}
    H_v = \begin{cases} \cos^2\eta \left( \sqrt{{R_\text{E}}^2 + \frac{{R_\text{S}}^2-{R_\text{E}}^2}{\cos^2\eta}}-R_\text{E} \right), & 0 \leq \eta < \frac{\pi}{2}, \\ 0, & \eta=\frac{\pi}{2}. \end{cases}
    \label{Formula:height}
\end{equation}
Then, the maximum distance from a service SAP to the typical UT is $\frac{H_v}{\cos\eta}$ for $0 \leq \eta < \frac{\pi}{2}$, and $\sqrt{{R_\text{S}}^2-{R_\text{E}}^2}$ for $\eta=\frac{\pi}{2}$.   

Finally, due to the change of $\eta$, each SAP has a service area on the earth's surface. The average number of UTs within this service area is written as \cite[Lemma 3]{li2024analytical}
\begin{equation}
    \left|\Phi_{l}^\text{U}\right|_\text{avg} = 2\pi R_\text{E}\lambda_\text{U} \begin{aligned}[t]
    &\left(R_\text{E}-\frac{{R_\text{E}}^2}{R_\text{S}} \sin^2\eta \right.\\
    &\left. - \frac{R_\text{E}\sqrt{{R_\text{S}}^2-{R_\text{E}}^2\sin^2\eta}\cos\eta}{R_\text{S}} \right).
    \end{aligned}
    \label{Formula:number_of_UTs}
\end{equation}

\subsection{Analytical Expressions}
The coverage probability is determined by the signal-to-interference-plus-noise ratio (SINR) of the message received by a typical UT and a threshold $\gamma_\text{th}$. When the SINR is above the threshold $\gamma_\text{th}$, this UT can successfully decode the data. 
Following the expression for the signals received by the $k$-th UT in (\ref{Formula:received_signal_at_UT}), the corresponding SINR is written as
\begin{equation}
    \mathrm{SINR}_k 
    = \frac{{{\left| \mathrm{DS}_k \right|}^{2}}} {\sum\limits_{k'\ne k} {\left| \mathrm{UI}_{kk'} \right|}^{2} + {\left| \mathrm{ISI}_{k} \right|}^{2} + {{\sigma }^{2}}},
\end{equation}
where 
\begin{equation} \nonumber
    \mathrm{DS}_k = \sum\limits_{l\in \Phi _{k}^{\text{Ser}}}{\sqrt{\rho _{lk}^\text{d}G_{\text{ml}}}{{\beta }_{lk}}^{\frac{1}{2}}{{h}_{lk}}\frac{\hat{h}_{lk}^{*}}{\left| {{\hat{h}}_{lk}} \right|}} q_k,
\end{equation}
\begin{equation} \nonumber
    \mathrm{UI}_{kk'} = \sum\limits_{l\in \Phi _{k}^{\text{Ser}}}{\sqrt{\rho _{lk'}^\text{d}G_{\text{ml}}}{{\beta }_{lk}}^{\frac{1}{2}}{{h}_{lk}}\frac{\hat{h}_{lk'}^{*}}{\left| {{\hat{h}}_{lk'}} \right|}} q_{k'},
\end{equation}
\begin{equation} \nonumber
    \mathrm{ISI}_{k} = \sum\limits_{l\in \Phi _{k}^{\text{Int}}}{\sqrt{{{\rho }_{d}}G_{\text{sl}}}{{\beta }_{lk}}^{\frac{1}{2}}{{h}_{lk}}} q_k,
\end{equation}
represent the strength of the DS, the interference caused by the $k'$-th user (UI), and the ISI, respectively.
By inserting $\rho_{lk}^\text{d}=\frac{\rho_\text{d}}{\left|\Phi_{l}^\text{U}\right|}$ for uniform power allocation, the coverage probability of the typical UT is represented as
\begin{equation}
\begin{aligned}
    &\mathbb{P}_{\text{cov}}\left(\gamma_\text{th};
    \lambda_\text{S},\lambda_\text{U},R_\text{S},\tau_\text{p},\rho_\text{d},m\right) \\
    &= \mathbb{P}\left\{\mathrm{SINR}_k > \gamma_\text{th} \right\} \\
    &= \mathbb{P}\left\{ {{\left| \sum\limits_{l\in \Phi _{k}^\text{Ser}}{\frac{ {{\beta }_{lk}}^{\frac{1}{2}}{{h}_{lk}}\frac{\hat{h}_{lk}^{*}}{\left| {{\hat{h}}_{lk}} \right|} }{\sqrt{\left| \Phi _{l}^\text{U} \right|}} } \right|}^{2}}\ge {{\gamma }_\text{th}} \left( I_{k}^{\text{Ser}}+I_{k}^{\text{Int}} \right) +\frac{{{\gamma }_\text{th}}{{\sigma }^{2}}}{{{\rho }_\text{d}}G_{\text{ml}}} \right\}\\
    &\leq \mathbb{P}\left\{ {{\left| \sum\limits_{l\in \Phi _{k}^\text{Ser}}{\frac{ {{\beta }_{lk}}^{\frac{1}{2}}{{h}_{lk}}\frac{h_{lk}^{*}}{\left| {{h}_{lk}} \right|} }{\sqrt{\left| \Phi _{l}^\text{U} \right|}} } \right|}^{2}} \ge {{\gamma }_\text{th}} \left( I_{k}^{\text{Ser}}+I_{k}^{\text{Int}} \right) +\frac{{{\gamma }_\text{th}}{{\sigma }^{2}}}{{{\rho }_\text{d}}G_{\text{ml}}} \right\} \\
    &= \mathbb{P}\left\{ \left| \sum\limits_{l\in \Phi _{k}^\text{Ser}}{\frac{ {{\beta }_{lk}}^{\frac{1}{2}}{{h}_{lk}}\frac{h_{lk}^{*}}{\left| {{h}_{lk}} \right|} }{\sqrt{\left| \Phi _{l}^\text{U} \right|}} } \right|\ge \sqrt{{{\gamma }_\text{th}} \left( I_{k}^{\text{Ser}}+I_{k}^{\text{Int}} \right) +\frac{{{\gamma }_\text{th}}{{\sigma }^{2}}}{{{\rho }_\text{d}}G_{\text{ml}}}} \right\},
\end{aligned}
\label{Formula_Coverage_Probability_def}
\end{equation}
where ${{I}_{k}^\text{Ser}}=\sum\limits_{k'\ne k}{{{\left| \sum\limits_{l\in \Phi _{k}^\text{Ser}}{\frac{1}{\sqrt{\left| \Phi _{l}^\text{U} \right|}}{{\beta }_{lk}}^{\frac{1}{2}}{{h}_{lk}}\frac{\hat{h}_{lk'}^{*}}{\left| {{\hat{h}}_{lk'}} \right|}} \right|}^{2}}}$, 
$I_{k}^{\text{Int}}= \frac{G_\text{sl}}{G_\text{ml}} {{\left| \sum\limits_{l\in \Phi _{k}^{\text{Int}}}{{{\beta }_{lk}}^{\frac{1}{2}}{{h}_{lk}}} \right|}^{2}}$.
We consider the perfect channel estimation in deriving an upper bound of the DSS.
To derive the coverage probability, we need to compute both the distribution of DSS $S_k$ and the average interference power.

\textcolor{black}{
Pilot-based channel estimation is included to reflect the complete cell-free massive MIMO framework. However, for analytical tractability, we assume perfect CSIT in subsequent derivations. 
This assumption is supported by numerical observations in the following Fig. \ref{fig:Fig_pilotnbr} showing that, with sufficient pilot resources, the estimation error becomes negligible, which makes perfect CSIT a reasonable assumption.
}

\captionsetup{font={scriptsize}}
\begin{figure}[tp]
\begin{center}
\setlength{\abovecaptionskip}{+0.4cm}
\setlength{\belowcaptionskip}{-0.4cm}
\centering
  \includegraphics[width=3.0in, height=1.8in]{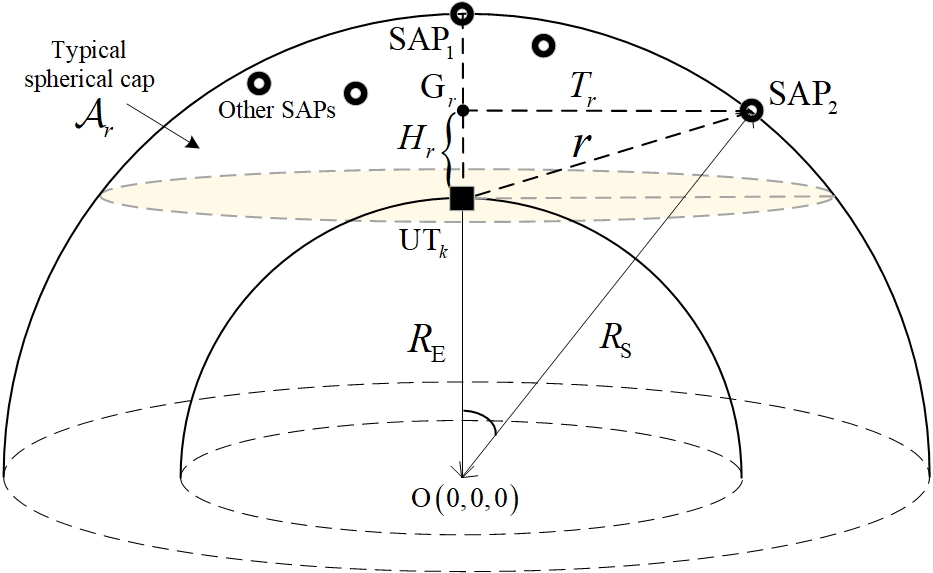}
\renewcommand\figurename{Fig.}
\caption{\scriptsize Statistical Properties for distance relationships.}
\label{fig:Spherical_cap}
\end{center}
\vspace{-5mm}
\end{figure}

\begin{proposition} 
Based on the perfect CSIT, the CCDF of the DSS received by the typical UT is given in (\ref{Formula:CCDF_DS}) \textcolor{black}{at the top of the next page}, where $s=\frac{A+i2\pi c}{2\sqrt{{{\left| \Phi _{l}^{\rm{U}} \right|}_{\rm{avg}}}} x}$.
\begin{figure*}[ht]
\setlength{\abovecaptionskip}{-1.5cm}
\setlength{\belowcaptionskip}{-0.5cm}
\normalsize
\begin{equation}
\begin{aligned}
  &\mathbb{P} \left\{ S_k \ge x \right\} \approx 1 - \frac{2^{-B}\rm{exp}\left( \frac{A}{2}\right)}{\sqrt{{{\left| \Phi _{l}^\text{U} \right|}_{\text{avg}}}} x} \sum_{b=0}^{B}\left( \tbinom{B}{b}\sum_{c=0}^{C+b}\frac{(-1)^b}{D_c} \rm{Re} \left[ \frac{1}{\textit{s}}\rm{exp}\left( -2\pi\lambda_S\frac{R_S}{R_E}\Theta(R_S,\eta,m) \right) \right]\right), {\;} {\rm{where}}\\ 
  & \Theta(R_\text{S},\eta,m) = \int_{r_{\text{S},\text{min}}}^{r_{\text{S},\text{max}}} \left\{ 1 - \frac{\Gamma\left(2m\right)}{\Gamma\left(m\right)2^{2m-1}} \left[ \frac{\sqrt{\pi}}{\Gamma(\frac{1}{2}+m)} \Phi\left( m,\frac{1}{2};\frac{s^2 r^{-\alpha}}{4m} \right) - \frac{s\sqrt{\frac{r^{-\alpha}\pi}{\Gamma(m)}}}{\Gamma(m)} \Phi\left(\frac{1+2m}{2},\frac{3}{2};\frac{s^2 r^{-\alpha}}{4m}\right) \right] \right\} dr. \\
\end{aligned}
\label{Formula:CCDF_DS}
\end{equation}
\hrulefill
\end{figure*}
\end{proposition}
\begin{proof}

First, DSS $S_k$ can be calculated as follows.
\begin{equation}
\begin{aligned}
    S_k 
    & = \sum\limits_{l\in \Phi _{k}^\text{Ser}}{\frac{1}{\sqrt{\left| \Phi _{l}^\text{U} \right|}}{{\beta }_{lk}}^{\frac{1}{2}}{{h}_{lk}}\frac{h_{lk}^{*}}{\left| {{h}_{lk}} \right|}} \\
    &\approx \frac{1}{\sqrt{{{\left| \Phi _{l}^\text{U} \right|}_{\text{avg}}}}}\sum\limits_{l\in \Phi _{k}^\text{Ser}}{{{\beta }_{lk}}^{\frac{1}{2}}{{h}_{lk}}\frac{h_{lk}^{*}}{\left| {{h}_{lk}} \right|}} \\
    &\approx \frac{1}{\sqrt{{{\left| \Phi _{l}^\text{U} \right|}_{\text{avg}}}}}\sum\limits_{l\in \Phi _{k}^\text{Ser}}{{{\beta }_{lk}}^{\frac{1}{2}}\left| {{h}_{lk}} \right|}.
\end{aligned}
\label{Formula:S_k}
\end{equation}
Denote ${{\hat{S}}_{k}}=\sum\limits_{l\in \Phi _{k}^\text{Ser}}{{{\beta }_{lk}}^{\frac{1}{2}}\left| {{h}_{lk}} \right|}$.
The PDF of ${\hat{S}}_{k}$ is obtained from its Laplace transform, which is
\begin{equation}
\begin{aligned}
    \mathcal{L}_{{\hat{S}}_{k}}(s) 
    & = \mathbb{E}\left\{ e^{-s {{\hat{S}}_{k}}} \right\} \\
    & = \int_{0}^{\infty}e^{-st}f_{{\hat{S}}_{k}}(t)dt,
\end{aligned}
\end{equation}
where $f_{{\hat{S}}_{k}} (t) = \mathcal{L}^{-1}\left\{ \mathcal{L}_{{\hat{S}}_{k}}(s) \right\}$,
and we also have
\begin{equation}
    \begin{aligned}
       \mathcal{L}_{{\hat{S}}_{k}}(s) 
       &= \mathbb{E}_{h} \left\{\mathrm{exp}(-s {{\hat{S}}_{k}})\right\} \\
       &= \mathbb{E}_{h} \left\{\prod_{l\in\Phi_{k}^\text{Ser}}\mathbb{E}_h \left\{ \mathrm{exp} \left( -s{\beta_{lk}}^{1/2} \left|h_{lk}\right| \right) \right\} \right\},
    \end{aligned}
    \label{Formula:LT_DS_power_1}
\end{equation}
where $\mathbb{E}_h \left\{ \mathrm{exp} \left( -s{\beta_{lk}}^{1/2} \left|h_{lk}\right| \right) \right\}$ needs to be first calculated.

As $\left| h_{lk} \right| \sim \mathrm{Nakagami}(m,\Omega)$, 
whose PDF is $f_{|h_{lk}|}(x) = \frac{2m^m}{\Gamma(m)}x^{2m-1}e^{-mx^2}$, the above expectation can be derived as 
\begin{align}
    &\mathbb{E}_h \left\{ \mathrm{exp} \left( -s{\beta_{lk}}^{1/2} \left|h_{lk}\right| \right) \right\} \nonumber \\
    &= \int_{0}^{\infty}e^{-s{\beta_{lk}}^{1/2}x} f_{|h_{lk}|}(x) dx \nonumber \\
    &= \int_{0}^{\infty}e^{-s{\beta_{lk}}^{1/2}x}\frac{2m^m}{\Gamma(m)}x^{2m-1}\mathrm{exp}(-mx^2)dx \nonumber \\
    &= \frac{2m^m}{\Gamma(m)} \int_{0}^{\infty}x^{2m-1}e^{-mx^2-s{\beta_{lk}}^2 x}dx \nonumber \\
    &\stackrel{(a)}{=} \frac{\Gamma(2m)}{\Gamma(m)2^{m-1}} e^{\frac{s^2\beta_{lk}}{8m}} D_{-2m}\left(s\sqrt{\frac{\beta_{lk}}{2m}}\right) \nonumber \\
    &\stackrel{(b)}{=} \frac{\Gamma(2m)}{\Gamma(m)2^{2m-1}} \begin{aligned}[t]
    &\left\{ \frac{\sqrt{\pi}}{\Gamma \left(\frac{1}{2}+m\right)} \Phi\left(m,\frac{1}{2};\frac{s^2 \beta_{lk}}{4m}\right) \right.\\
    &\left. - \frac{s\sqrt{\frac{\beta_{lk}\pi}{m}}}{\Gamma(m)}\Phi\left(\frac{1}{2}+m,\frac{3}{2};\frac{s^2\beta_{lk}}{4m}\right) \right\},
    \end{aligned}
\label{Formula:LT_DS_power_E}
\end{align}
where $(a)$ follows from a combinations of exponentials of more complicated arguments and powers in \cite[Equation 3.462.1]{zwillinger2007table}, $D_{-2m}\left(s\sqrt{\frac{\beta_{lk}}{2m}}\right)$ is written as (\ref{Formula:D_-2m}) 
\textcolor{black}{at the top of the next page, }
and $(b)$ follows from the parabolic cylinder function in \cite[Equation 9.240.1]{zwillinger2007table}.
\begin{figure*}[ht!]
\setlength{\abovecaptionskip}{-1.5cm}
\setlength{\belowcaptionskip}{-0.5cm}
\normalsize
\vspace{-3mm}
\begin{equation}
\begin{split}
    \begin{aligned}
        {{D}_{-2m}}\left( s\sqrt{\frac{{{\beta }_{lk}}}{2m}} \right)
        = {{2}^{-m}}\exp \left( -\frac{{{s}^{2}}{{\beta }_{lk}}}{8m} \right)\left\{ \frac{\sqrt{\pi }}{\Gamma \left( \frac{1+2m}{2} \right)}\Phi \left( m,\frac{1}{2};\frac{{{s}^{2}}{{\beta }_{lk}}}{4m} \right)-\frac{s\sqrt{\frac{{{\beta }_{lk}}\pi }{m}}}{\Gamma \left( m \right)}\Phi \left( \frac{1+2m}{2},\frac{3}{2};\frac{{{s}^{2}}{{\beta }_{lk}}}{4m} \right) \right\}.
    \end{aligned}
    \label{Formula:D_-2m}
\end{split}
\end{equation}
\hrulefill
\end{figure*}

To illustrate the relationship between the satellite-terrestrial distance and other parameters on two spheres, Fig. \ref{fig:Spherical_cap} is provided, where the previous distance symbol $H_{v}$ is re-represented by $H_r$ based on the SAP-to-UT distance $r$.
The area of the typical spherical cap is $\left| \mathcal{A} \right| = 2\pi(R_\text{S}-R_\text{E})R_\text{S}$.

Let the distance from a randomly-selected SAP on the sphere $S$ to the typical UT be $r$. In Triangle $2$ by SAP$_2$, point $\text{G}_r$ and original point $\text{O}$, the height of $H_r$ can be derived by using Pythagorean theorem ${R_\text{S}}^2 = (R_\text{E} + H_r)^2+{T_r}^2$, 
and in Triangle $1$ formed by SAP$_2$, point $\text{G}_r$ and UT$_k$, distance $r$ can be written as $r^2 = {T_r} ^2 + {H_r}^2$.
By combining the above two formulas, the expression for $H_r$ on $r$ is given by
$H_r = \frac{{R_\text{S}}^2-{R_\text{E}}^2-r^2}{2 R_\text{E}}$.
Then, the typical spherical cap can be further represented as 
\begin{equation}
    \begin{aligned}
        \left| \mathcal{A}_r \right| &= 2\pi(R_\text{S}-R_\text{E}-H_r)R_\text{S} \\
        &= 2\pi\left[ R_\text{S}-R_\text{E}-\frac{\left( {R_\text{S}}^2 - {R_\text{E}}^2 \right)-r^2}{2R_\text{E}} \right] R_\text{S}.
    \end{aligned}
\end{equation}
The derivation of $\mathcal{A}_r$ on $r$ is written as
\begin{equation}
    \frac{\partial \left| \mathcal{A}_r \right|}{\partial r} = 2\frac{R_\text{S}}{R_\text{E}}\pi r.
    \label{Formula:derivation_A}
\end{equation}
Then, by inserting (\ref{Formula:LT_DS_power_E}) back into (\ref{Formula:LT_DS_power_1}), the Laplace transform of the desired signal power can be further written as
\begin{equation}
    \begin{aligned}
        & \mathcal{L}_{{\hat{S}}_{k}}(s) \\
        &= \mathbb{E}_{h}\left\{\prod_{l\in\Phi_{k}^\text{Ser}}\frac{\Gamma(2m)}{\Gamma(m)2^{2m-1}} e^{\frac{s^2\beta_{lk}}{8m}} D_{-2m}\left(s\sqrt{\frac{\beta_{lk}}{2m}}\right) \right\} \\
        &\stackrel{(a)}{=} \mathrm{exp} \left[ -\lambda_S\int_{v  \in\mathcal{A}_r} 1- 
        \frac{\Gamma(2m)}{\Gamma(m)2^{m-1}}e^{\frac{s^2 v^{-\alpha}}{8m}} \right.\\ 
        & \qquad \qquad \qquad \qquad \qquad \qquad \left.\times D_{-2m}\left(s\sqrt{\frac{\beta_{lk}}{2m}}\right)d v \right],
    \end{aligned}
    \label{Formula:LT_DS_power_2}
\end{equation}
where $(a)$ follows from the probability generating functional (PGFL) of the PPP.

By inserting (\ref{Formula:derivation_A}) into (\ref{Formula:LT_DS_power_2}), we can further derive $\mathcal{L}_{{\hat{S}}_{k}}(s)$ as in (\ref{Formula:LT_DS_power_3}), which is \textcolor{black}{at the top of the next page}.
\begin{figure*}[ht]
\setlength{\abovecaptionskip}{-1.5cm}
\setlength{\belowcaptionskip}{-0.5cm}
\normalsize
\vspace{-4mm}
\begin{equation}
    \begin{aligned}
      \mathcal{L}_{{\hat{S}}_{k}}(s) 
      &= \mathrm{exp} \left\{ -2\pi\lambda_\text{S} \frac{R_\text{S}}{R_\text{E}} \int_{r_{\text{S},\rm{min}}}^{r_{\text{S},\rm{max}}} \left[ 1 - \frac{\Gamma(2m)}{\Gamma(m)2^{m-1}} e^{\frac{s^2 r^{-\alpha}}{8m}} D_{-2m}\left(s\sqrt{\frac{\beta_{lk}}{2m}}\right) \right]rdr \right\} \\
      & = {\rm{exp}}\left\{ { - 2\pi {\lambda_\text{S}}\frac{{{R_\text{S}}}}{{{R_\text{E}}}}\int_{{r_{\text{S},{\rm{min}}}}}^{{r_{\text{S},{\rm{max}}}}} r } \right.\\
      & {\quad\quad\quad\quad\quad} \left. { \times \left[ {1 - \frac{{\Gamma \left( {2m} \right)}}{{\Gamma \left( m \right){2^{2m - 1}}}}\left( {\frac{{\sqrt \pi  }}{{\Gamma \left( {\frac{1}{2} + m} \right)}}\Phi \left( {m,\frac{1}{2};\frac{{{s^2}{\beta _{lk}}}}{{4m}}} \right) - \frac{{s\sqrt {\frac{{{\beta _{lk}}\pi }}{m}} }}{{\Gamma \left( m \right)}}\Phi \left( {\frac{1}{2} + m,\frac{3}{2};\frac{{{s^2}{\beta _{lk}}}}{{4m}}} \right)} \right)} \right]dr} \right\}.
    \end{aligned}
    \label{Formula:LT_DS_power_3}
\end{equation}
\hrulefill
\begin{equation}
\begin{aligned}
  F_{{\hat{S}}_{k}}(x) 
  &\approx \frac{2^{-B}\mathrm{exp}\left(\frac{A}{2}\right)}{x} \sum_{b=0}^{B} \left( \dbinom{B}{b} \sum_{c=0}^{C+b}\frac{(-1)^{c}}{D_c} \mathrm{Re}\left[ \frac{\mathcal{L}_{{\hat{S}}_{k}}(s)}{s} \right] \right) \\
  &= \frac{2^{-B}\mathrm{exp}\left(\frac{A}{2}\right)}{x} \sum_{b=0}^{B} \left( \dbinom{B}{b} \sum_{c=0}^{C+b}\frac{(-1)^{c}}{D_c} \mathrm{Re}\left[ \frac{1}{s} \mathrm{exp} \left( -2\pi\lambda_\text{S} \frac{R_\text{S}}{R_\text{E}}\Theta(R_\text{S},\eta,m) \right) \right] \right).
\end{aligned}
    \label{Formula:F_S_k_2}
\end{equation}
\hrulefill
\vspace{-4mm}
\end{figure*}
To derive the CDF of ${{\hat{S}}_{k}}$, we have an integral over its PDF and adopt the inverse Laplace transform method for the PDF \cite{shang2017enabling}, where it can be calculated by Bromwich contour integral method.
Moreover, it can also be estimated using summation forms of integrals with a controllable error \cite{cohen2007numerical}. 
We derive the CDF of ${{\hat{S}}_{k}}$, as
\begin{equation}
\begin{aligned}
    F_{{\hat{S}}_{k}} (x) 
    &= \mathbb{P} \left\{ {{\hat{S}}_{k}} \leq x \right\} \\
    &= \int_{0}^{x}f_{{\hat{S}}_{k}}(t)dt \\
    &= \int_{0}^{x} \mathcal{L}^{-1}\left\{ \mathcal{L}_{{\hat{S}}_{k}} (s) \right\} dt \\
    &\stackrel{(a)}{=} \int_{0}^{x} \frac{1}{2\pi i} \lim_{T\to\infty} \int_{c-iT}^{c+iT} e^{sT} \mathcal{L}_{F_{{\hat{S}}_{k}}}(s) dsdt \\
    &\stackrel{(b)}{=} \frac{1}{2\pi i} \lim_{T\to\infty} \int_{c-iT}^{c+iT} e^{sx} \frac{\mathcal{L}_{{\hat{S}}_{k}}(s)}{s} ds,
\end{aligned}
\label{Formula:F_S_k_1}
\end{equation}
where $(a)$ holds true according to the definition of inverse Laplace transform, and $(b)$ is obtained based on $\mathcal{L}_{F_{{\hat{S}}_{k}}}=\frac{\mathcal{L}_{{\hat{S}}_{k}}}{s}$ in probability theory.
Following \cite{zhou2017power}, (\ref{Formula:F_S_k_1}) can be approximated by a finite sum which is expressed as in (\ref{Formula:F_S_k_2}) \textcolor{black}{at the top of the page},
where $s=\frac{A+i2\pi c}{2x}$, and $\mathcal{L}_{{\hat{S}}_{k}}(s)$ is obtained in (\ref{Formula:LT_DS_power_3}). $A$, $B$, and $C$ are positive parameters used to adjust the accuracy of the approximation. We also let $D_c=1$ for $c=1,2,...,C+b$ and $D_c=2$ for $c=0$. 
Considering the approximation of ${{\hat{S}}_{k}}\approx \sqrt{{{\left| \Phi _{l}^\text{U} \right|}_{\text{avg}}}}{{S}_{k}}$ in (\ref{Formula:S_k}) and the CDF of ${{\hat{S}}_{k}}$, i.e., $F_{{\hat{S}}_{k}} (x) = \mathbb{P} \left\{ {{\hat{S}}_{k}} \leq x \right\}$ in (\ref{Formula:F_S_k_2}), 
the CCDF of DSS $S_k$, i.e., $\mathbb{P}\left\{ {{S}_{k}}\ge x \right\}$, is given in (\ref{Formula:CCDF_DS}), completing the proof.
\end{proof}

\begin{proposition} 
    The average MUI received by the typical UT is given by
    \begin{equation}
        \begin{aligned}
        &I_{k,\rm{avg}}^{\rm{Ser}} 
        = \mathbb{E}\left\{ I_{k}^{\rm{Ser}} \right\} \\
        &= \frac{2\pi {{\lambda }_{\rm{S}}}{{R}_{\rm{S}}}\Omega }{\left( 2-\alpha  \right){{R}_{\rm{E}}}} \frac{{{\left| \Phi _{l}^{\rm{U}} \right|}_{\rm{avg}}}-1}{{{\left| \Phi _{l}^{\rm{U}} \right|}_{\rm{avg}}}} \left[ {{r}_{\rm{S},\max }}^{2-\alpha }-{{r}_{\rm{S},\min }}^{2-\alpha } \right].
        \end{aligned}
        \label{Formula:MUI}  
    \end{equation}
    where $\left|\Phi_{l}^{\rm{U}}\right|_{\rm{avg}}$ is provided in (\ref{Formula:number_of_UTs}),
    and $I_{k}^{\rm{Ser}}$ is provided and expressed in (\ref{Formula_Coverage_Probability_def}). 
    Then, the average ISI received by the typical UT is given by
    \begin{equation}
        \begin{aligned}
        I_{k,\rm{avg}}^{\rm{Int}} 
        &= \mathbb{E}\left\{ I_{k}^{\rm{Int}} \right\} \\
        &= \frac{G_{\rm{sl}}}{G_{\rm{ml}}} \frac{2\pi {{\lambda }_{\rm{S}}}{{R}_{\rm{S}}}\Omega }{\left( 2-\alpha  \right){{R}_{\rm{E}}}}\left[ {{r}_{\max }}^{2-\alpha }-{{r}_{\rm{S},\max }}^{2-\alpha } \right],
        \end{aligned}
        \label{Formula:ISI} 
    \end{equation}    
\end{proposition}
\begin{proof}
The average MUI is represented as
\begin{align}
        I_{k,\text{avg}}^{\text{Ser}}
        & \overset{\left( a \right)}{\mathop{\textcolor{black}{=}}}\, \mathbb{E}\left\{ I_{k}^{\text{Ser}} \right\} \nonumber \\
        & \overset{\left( b \right)}{\mathop{=}}\,\mathbb{E}\left\{ \sum\limits_{k'\ne k}{{{\left| \sum\limits_{l\in \Phi _{k}^{\text{Ser}}}{\frac{1}{\sqrt{\left| \Phi _{l}^\text{U} \right|}}{{\beta }_{lk}}^{\frac{1}{2}}{{h}_{lk}}} \right|}^{2}}} \right\} \nonumber \\ 
        & \overset{\left( c \right)}{\mathop{\textcolor{black}{\approx}}}\,\mathbb{E}\left\{ \sum\limits_{k'\ne k}{\sum\limits_{l\in \Phi _{k}^{\text{Ser}}}{\frac{1}{\left| \Phi _{l}^\text{U} \right|}{{\beta }_{lk}}{{\left| {{h}_{lk}} \right|}^{2}}}} \right\} \nonumber \\ 
        & \overset{\left( d \right)}{\mathop{\textcolor{black}{\approx}}}\,\mathbb{E}\left\{ \sum\limits_{l\in \Phi _{k}^{\text{Ser}}}{\frac{1}{\left| \Phi _{l}^\text{U} \right|}\sum\limits_{\begin{smallmatrix} 
        k'\ne k \\ 
        k'=1 
        \end{smallmatrix}}^{\left| \Phi _{l}^\text{U} \right|}{{{\beta }_{lk}}{{\left| {{h}_{lk}} \right|}^{2}}}} \right\} \nonumber \\ 
        & \overset{\left( e \right)}{\mathop{\textcolor{black}{=}}}\,\frac{2\pi {{\lambda }_{\text{S}}}{{R}_{\text{S}}}\Omega }{{{R}_{\text{E}}}} \mathbb{E}\left\{\frac{{{\left| \Phi _{l}^\text{U} \right|}}-1}{{{\left| \Phi _{l}^\text{U} \right|}}} \right\} \int_{{{r}_{\text{S},\min }}}^{{{r}_{\text{S},\max }}}{r\cdot {{r}^{-\alpha }}dr} \nonumber \\
        & \overset{\left( f \right)}{\mathop{\textcolor{black}{=}}}\,\frac{2\pi {{\lambda }_{\text{S}}}{{R}_{\text{S}}}\Omega }{{{R}_{\text{E}}}}\frac{{{\left| \Phi _{l}^\text{U} \right|}_{\text{avg}}}-1}{{{\left| \Phi _{l}^\text{U} \right|}_{\text{avg}}}}\int_{{{r}_{\text{S},\min }}}^{{{r}_{\text{S},\max }}}{r\cdot {{r}^{-\alpha }}dr},
    \label{Formula:Interference part proof}  
\end{align}
where 
\textcolor{black}{$(a)$ is based on the definition of expectation over interference \cite{shang2024multi}}
, $(b)$ is due to the nature of normalization of $\frac{\hat{h}_{lk'}^{*}}{\left| {{\hat{h}}_{lk'}} \right|}$ to $1$, $(c)$ follows an approximation in \cite{nigam2014coordinated} that sums up the interference from the remaining, non-service LEO satellites, which considers the effects of MUI resulting from the transmitted signal from $l$-th SAP to all its served UTs, $(d)$ comes from the average number of UTs in the coverage of an SAP, and the exchange of sum terms, $(e)$ is due to the count of the total number of UTs in the coverage of an $l$-th SAP, except the $k$-th UT, which is currently being cooperatively served, and $(f)$ is obtained from the definition of expectation and considering the average total number of UTs in the coverage of an $l$-th SAP, which can be reasonable in the average sense.
Similarly, the ISI can be derived as
\begin{equation}
    \begin{aligned}
        I_{k,\text{avg}}^{\text{Int}}
        & =\mathbb{E}\left\{ I_{k}^{\text{Int}} \right\} \\
        & =\mathbb{E}\left\{ \frac{G_\text{sl}}{G_\text{ml}} \sum\limits_{l\in \Phi _{k}^{\text{Int}}}{{{\beta }_{lk}}{{\left| {{h}_{lk}} \right|}^{2}}} \right\} \\ 
        & =\frac{G_\text{sl}}{G_\text{ml}} \frac{2\pi {{\lambda }_{\text{S}}}{{R}_{\text{S}}}\Omega }{{{R}_{\text{E}}}}\int_{{{r}_{\text{S},\max }}}^{{{r}_{\max }}}{r\cdot {{r}^{-\alpha }}dr}.
    \end{aligned}
\end{equation}
Finally, \textcolor{black}{analytical} expressions are obtained by inserting (\ref{Formula:derivation_A}) and using Campbell's theorem, completing the proof.
\end{proof}

\textcolor{black}{
\begin{remark}
    (Impact of User Number on average MUI). 
    Due to the coefficient $\frac{{{\left| \Phi _{l}^{\rm{U}} \right|}_{\rm{avg}}}-1}{{{\left| \Phi _{l}^{\rm{U}} \right|}_{\rm{avg}}}} = 1 - \frac{1}{{{\left| \Phi _{l}^{\rm{U}} \right|}_{\rm{avg}}}}$ in \eqref{Formula:MUI}, it can be observed that the average MUI at a UT increases as the number of users increases, i.e., the denominator becomes larger. However, this coefficient tends towards a saturation value $1$. 
    This implies that when the user density is low, the MUI increases rapidly; when the user density is high, the MUI tends to stabilize, and the increase in capacity will show a marginal diminishing effect in the high user density range.
\end{remark}
}

\textcolor{black}{
\begin{remark}
    (Impact of Main-Lobe-Side-Lobe Gain Disparity on ISI). 
    The ISI in LEO satellite CF-mMIMO networks is largely dependent on the ratio of main-lobe gain to side-lobe gain, i.e., coefficient $\frac{G_\text{sl}}{G_\text{ml}}$ in \eqref{Formula:ISI}. Specifically, both the increase in side-lobe gain and the decrease of main-lobe gain will lead to a general increased ISI.
    This implies that while the main-lobe ensures reliable signal reception for intended UTs, side-lobe radiation may accumulate as dominant interference sources for UTs located in neighboring beams. 
    Therefore, the antenna pattern design directly determines the intensity of network interference. It can be more effective to increase the ratio of $\frac{G_\text{sl}}{G_\text{ml}}$ than simply increasing the number of SAPs without antenna gain considerations.
\end{remark}
}

\begin{theorem} 
    The coverage probability on the Nakagami-$m$ fading channel is given in (\ref{Formula:Coverage Probability}) \textcolor{black}{at the top of the next page}, where $s=\frac{A+i2\pi c}{2 \sqrt{{{\left| \Phi _{l}^{\rm{U}} \right|}_{\rm{avg}}} \left[\gamma_{\rm{th}} \left( I_{k,\rm{avg}}^{\rm{Ser}}+I_{k,\rm{avg}}^{\rm{Int}} \right) + \frac{\gamma_{\rm{th}} \sigma^2}{\rho_{\rm{d}} G_{\rm{ml}}} \right]} }$, $I_{k,{\rm{avg}}}^{\rm{Ser}}$ and $I_{k,{\rm{avg}}}^{\rm{Int}}$ are given in Proposition 2.

\begin{figure*}[ht]
\setlength{\abovecaptionskip}{-1.5cm}
\setlength{\belowcaptionskip}{-0.5cm}
\normalsize
\begin{equation}
    \begin{aligned}
        & \mathbb{P}_k^{{\text{cov}}}({\gamma_\text{th}};{\lambda_\text{S}},{\lambda_\text{U}},{R_\text{S}},{\rho_\text{d}},{\eta},{m}) 
        \lesssim 1 - \frac{2^{-B}\rm{exp}\left( \frac{\textit{A}}{2}\right)\sum_{\textit{b}=0}^\textit{B}\left( \tbinom{\textit{B}}{\textit{b}}\sum_{\textit{c}=0}^\textit{C+b}\frac{(-1)^\textit{b}}{\textit{D}_\textit{c}} \rm{Re} \left[ \frac{1}{\textit{s}}\rm{exp}\left( -2\pi\lambda_S\frac{R_S}{R_E}\Theta(R_S,\eta,\textit{m}) \right) \right]\right)} {\sqrt{{{\left| \Phi _{l}^\text{U} \right|}_{\text{avg}}} \left[\gamma_\text{th} \left( I_{k,\text{avg}}^{\text{Ser}}+I_{k,\text{avg}}^{\text{Int}} \right) + \frac{\gamma_\text{th} \sigma^2}{\rho_\text{d} G_{\text{ml}}} \right]} }.
    \end{aligned}
\label{Formula:Coverage Probability}
\end{equation}
\hrulefill
\end{figure*}
\end{theorem}

\begin{proof}
    Due to the characterization of CCDF, by inserting the right-hand side in the last step of (\ref{Formula_Coverage_Probability_def}), which includes the interference signal power in Proposition 2 and the noise power $\sigma^2$, into the CCDF of the desired signal part shown in (\ref{Formula:CCDF_DS}) of Proposition 1, an approximate expression is obtained for the coverage probability of CF-mMIMO SatCom networks in (\ref{Formula:Coverage Probability}) \textcolor{black}{at the top of the next page}, which completes the proof.
\end{proof}

\begin{remark}
    This step is an approximation that assumes the aggregate interference exhibits concentration around its mean due to the large number of contributing interferers. This is justified under the Law of Large Numbers, which holds in our setting, where spatial randomness and statistical independence govern the distribution of interference sources. This approximation has been validated through extensive Monte Carlo simulations, where it is observed that the instantaneous interference variation is tightly centered around $\mathbb{E} \left[ I_{k}^{\text{Ser}} + I_{k}^{\text{Int}} \right]$.
\end{remark}

While the derived expressions in Section \uppercase\expandafter{\romannumeral3} above are not presented in a sufficiently concise form due to the inherent complexity of the geometry of the CF-mMIMO LEO SatCom network and the cooperative service mechanisms; they remain analytically structured in terms of key system parameters. This enables numerical evaluation to extract insights into performance trends and parameter sensitivities shown in the following section, which is essential for network designs.

\section{Simulation Results and Discussions}
In this section, we quantitatively investigate the downlink performance of the CF-mMIMO satellite network with respect to essential parameters, including dome angle, orbital altitude and number of SAPs, Nakagami fading parameter $m$, etc. 
The analytical results are provided based on the statistical properties and expressions in Sections \uppercase\expandafter{\romannumeral2} and \uppercase\expandafter{\romannumeral3}.
We perform Monte Carlo simulations to obtain the DSS, coverage probability, capacity, and other tentative indices for numerical and analytical verification. 
Unless otherwise noted, the values for the system parameters are summarized in Table \ref{tab:system parameters} according to \cite{shang2024multi,park2022tractable,abdelsadek2023broadband,3gpp.38.821}, and other parameters are based on different scenarios.
\textcolor{black}{
Note that the analytical results for coverage probability in the following figures are mainly based on \eqref{Formula:Coverage Probability}, and those for capacity are mainly based on \eqref{Formula:syst_cap_cf}, \eqref{Formula:syst_cap_Nearest}, \eqref{Formula:per_cap_cf}, \eqref{Formula:per_cap_Nearest}, which will be shown in the following subsections.
}

\begin{table}[t]
\begin{center}
\captionsetup{font={normalsize}}
\caption{System parameters}
\label{tab:system parameters}
\small
\begin{tabular}{cc}
\hline
Parameter                      & Value       \\ \hline
Radius of Earth $R_\text{E}$              & $6371.393$ km \\
Orbital altitude $H_\text{S}$             & $500$ km \\
Dome angle $\eta$    & $75^{\circ}$ \\
Density of CF-mMIMO SAPs $\lambda_\text{S}$  & $1 \times 10^{-5}$/km\textsuperscript{2}   \\
Density of served UTs  $\lambda_\text{U}$    & $3\times 10^{-6}$/km\textsuperscript{2}   \\
Carrier frequency $f_\text{c}$             & $2$ GHz \\
Transmit power of downlink data & $33$ dBm         \\
Transmit power of pilot symbol  & $30$ dBm         \\
Noise power  $\sigma^2$       & $-100$ dBm     \\
SAP antenna gain $G_{\text{ml/sl}}$     & $30$ dBi/$20$ dBi  \\ 
UT antenna gain $G_\text{r}$               & $0$ dBi        \\
Nakagami parameter $m$               & $2$       \\
Path-loss exponent $\alpha$    & $2$  \\\hline
\end{tabular}
\end{center}
\end{table}

\captionsetup{font={scriptsize}}
\begin{figure}[tp]
\begin{center}
\setlength{\abovecaptionskip}{+0.2cm}
\setlength{\belowcaptionskip}{-0.4cm}
\centering
  \includegraphics[width=3.4in, height=2.7in]{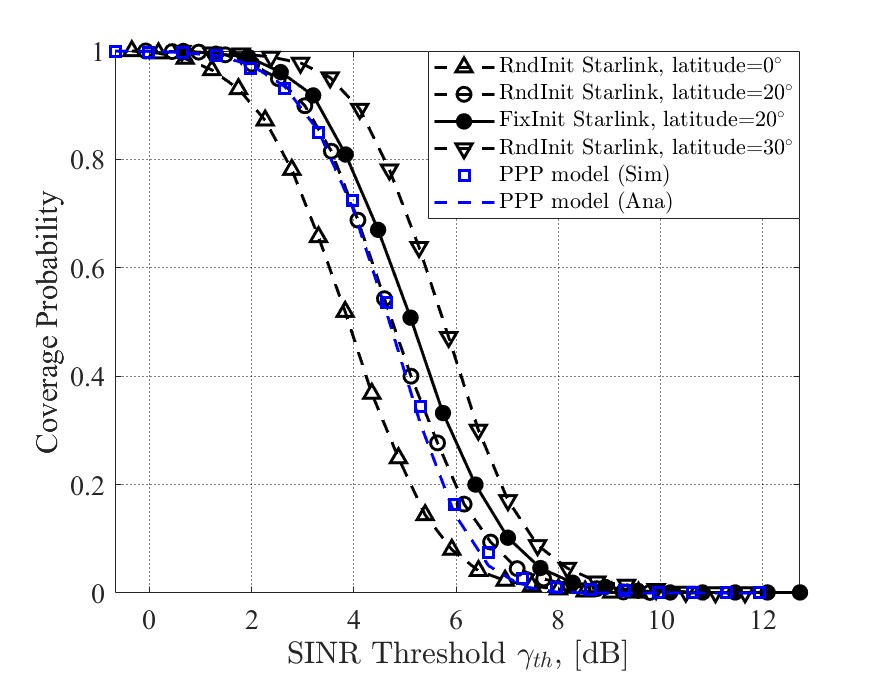}
\renewcommand\figurename{Fig.}
\caption{\scriptsize Comparison of coverage probability among PPP model, random initialized Starlink, and fixed initialized Starlink.}
\label{fig:Fig_PPPvsStarlink}
\end{center}
\end{figure}

\captionsetup{font={scriptsize}}
\begin{figure}[tp]
\begin{center}
\setlength{\abovecaptionskip}{+0.2cm}
\setlength{\belowcaptionskip}{-0.4cm}
\centering
  \includegraphics[width=3.4in, height=2.7in]{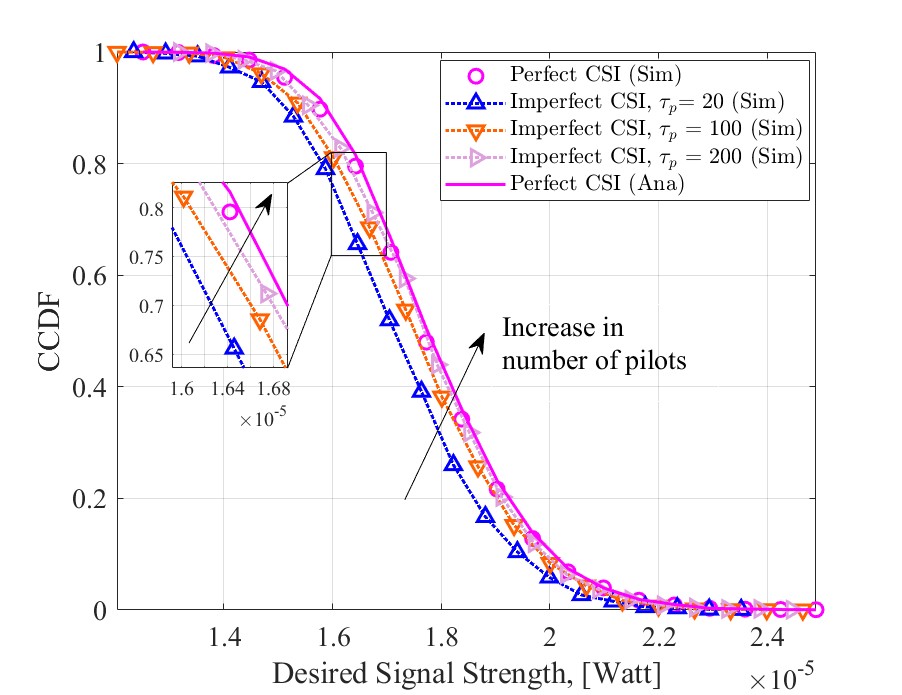}
\renewcommand\figurename{Fig.}
\caption{\scriptsize CCDF of desired signal strength, where $\eta=90^{\circ}$.}
\label{fig:Fig_pilotnbr}
\end{center}
\end{figure}

\subsection{Coverage Probability}

To validate the suitability of the PPP model for the studied network, we first compare the coverage probability under the PPP model, random-initial-state and fixed-initial-state Starlink constellation at an example orbital altitude of $500$ km \cite{jia2023analytic,al2021modeling}.
\textcolor{black}{
Note that the Starlink constellation is simulated with three typical inclination angles, i.e., $33^{\circ}$, $43^{\circ}$, $53^{\circ}$, each having $28$ orbital planes.}

\textcolor{black}{
Fig. \ref{fig:Fig_PPPvsStarlink} shows that when observed from a location at a latitude of $20^{\circ}$, the coverage gap between the PPP model and the fixed-initial-state Starlink constellation is more significant than that between the PPP model and the random-initial-state Starlink constellation. 
In addition, the observed discrepancies at three different latitudes primarily result from the uneven satellite distribution in Starlink constellations, with higher density at mid-latitudes and lower density near the equator. 
This observation reinforces the necessity of our findings. Despite its abstraction, the PPP-based analytical model effectively captures average performance trends and provides theoretical design insights.
}
This indicates that the PPP model can fit the actual constellation well with the random-initial-state Starlink constellation model. 
As the actual constellation is close to the random-initial-state Starlink constellation, it is reasonable to use the PPP model for approximation.
Moreover, it is assumed that the distribution of SAPs is not correlated with the PPP for the derivation of analytical results.

\captionsetup{font={scriptsize}}
\begin{figure}[tp]
\begin{center}
\setlength{\abovecaptionskip}{+0.2cm}
\setlength{\belowcaptionskip}{-0.4cm}
\centering
  \includegraphics[width=3.4in, height=2.7in]{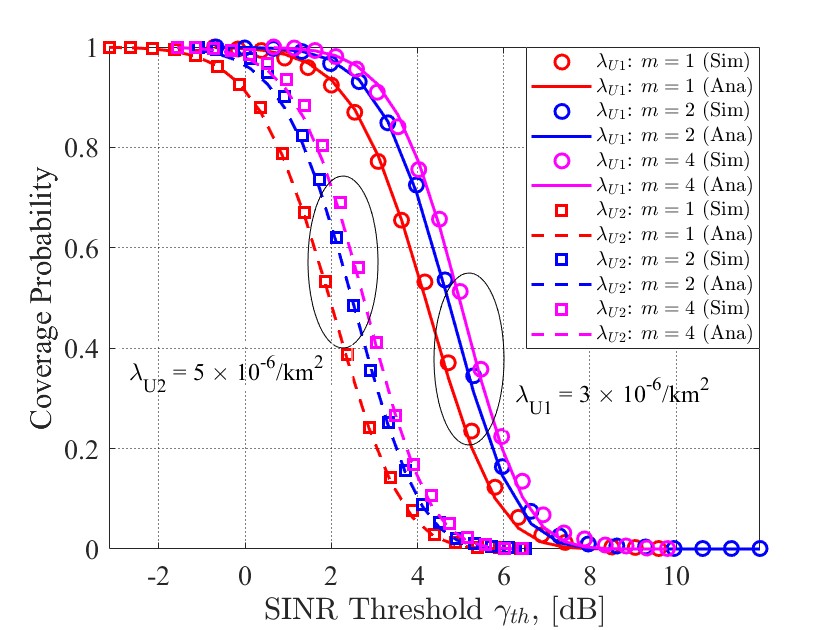}
\renewcommand\figurename{Fig.}
\caption{\scriptsize Coverage probability versus SINR thresholds for different Nakagami fading parameters $m$.}
\label{fig:Fig_m}
\end{center}
\end{figure}

\captionsetup{font={scriptsize}}
\begin{figure}[tp]
\begin{center}
\setlength{\abovecaptionskip}{+0.2cm}
\setlength{\belowcaptionskip}{-0.4cm}
\centering
  \includegraphics[width=3.4in, height=2.7in]{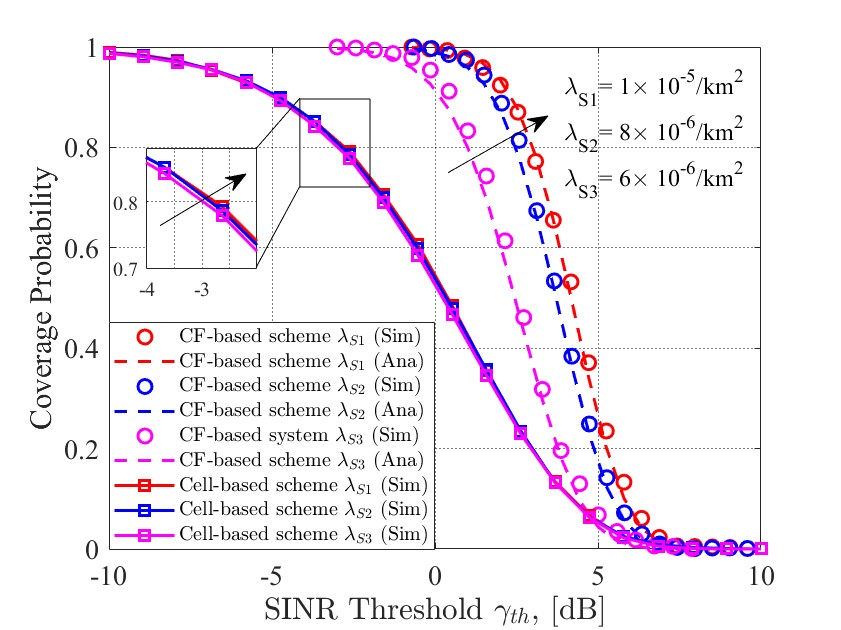}
\renewcommand\figurename{Fig.}
\caption{\scriptsize Coverage probability comparison between CF-based SatCom networks and Cell-based SatCom networks.}
\label{fig:Fig_cov_baseline}
\end{center}
\end{figure}

\captionsetup{font={scriptsize}}
\begin{figure}[tp]
\begin{center}
\setlength{\abovecaptionskip}{+0.2cm}
\setlength{\belowcaptionskip}{-0.4cm}
\centering
  \includegraphics[width=3.4in, height=2.7in]{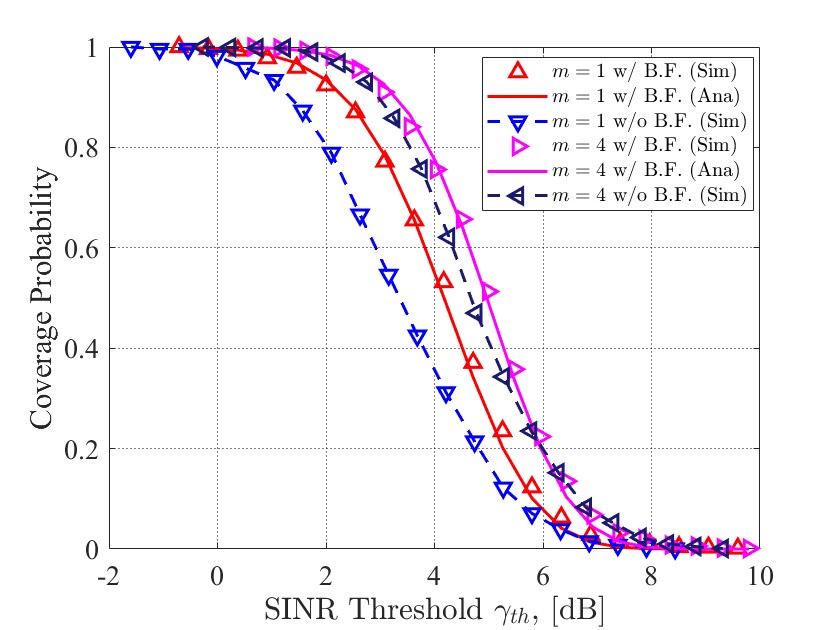}
\renewcommand\figurename{Fig.}
\caption{\scriptsize Coverage probability versus the SINR threshold with and without beamforming.}
\label{fig:Fig_BF}
\end{center}
\end{figure}

To understand how well the analytical DSS matches its simulation counterparts, we compare its CCDF with that under perfect and imperfect CSI in Fig. \ref{fig:Fig_pilotnbr}. The CCDF of DSS with imperfect CSI depends on the number of pilots during the uplink training stage. 
Specifically, when $\tau_\text{p}=20$, an apparent gap can be witnessed between the CCDF of DSS under imperfect CSI and that under perfect CSI. 
However, the gap reduces as more pilots are utilized, e.g., $\tau_\text{p}=100$, $\tau_\text{p}=200$. We can see that the increase in the pilots' number gives rise to a more accurate DSS, approximating one under perfect CSI. 
It can also be seen that the analytical DSS matches the simulation DSS well under perfect CSI. Thus, it is reasonable to use this analytical expression for further investigation. This also shows that by increasing the number of pilots, a more accurate estimated channel can be obtained.

\textcolor{black}{
Fig. \ref{fig:Fig_m} compares the coverage probability under different UT densities and Nakagami fading parameters $m$. The analytical expressions derived from Theorem 1 are shown to closely match the simulation results at different values of $m$, thus validating the precision of the theoretical analysis.
Focusing on the case where the UT density is $\lambda_\text{U} = 3\times 10^{-6}$/km$^2$, it is observed that the coverage probability improves as $m$ increases. Specifically, within the SINR threshold range of approximately $2$ dB to $8$ dB, the case with $m=4$ achieves the highest coverage probability, followed by the cases with $m=2$ and $m=1$, respectively. The performance gap between $m=4$ and $m=1$ reaches its maximum of around $0.25$ at $\gamma_\text{th}=5$ dB and gradually decreases beyond $\gamma_\text{th}=6$ dB. A similar but less evident trend is also evident between the $m=4$ and $m=2$ cases. 
This behavior persists when the UT density increases to $\lambda_\text{U}=5\times 10^{-6}$/km$^2$. In general, it is also revealed that higher UT densities lead to lower coverage probabilities due to increased aggregate MUI, regardless of the fading parameter $m$.
}

\textcolor{black}{
In Fig. \ref{fig:Fig_cov_baseline}, we present a comparative analysis between the proposed CF-based satellite architecture and the existing cell-based architecture, where each typical UT is served by its nearest satellite \cite{park2022tractable}. In the cell-based architecture, the coverage probability exhibits only marginal improvements as the satellite density increases. This is because a higher satellite density slightly increases the likelihood that a UT will be served by a closer satellite, thus improving the received signal strength.
In contrast, the CF-based satellite architecture demonstrates a more pronounced improvement in coverage probability with increasing satellite density. In particular, even with moderate densities such as $\lambda_\text{S2}$ and $\lambda_\text{S3}$, the CF-based scheme consistently outperforms the cell-based counterpart. Furthermore, when the satellite density is $\lambda_\text{S1}=1\times10^{-5}$/km$^{2}$, the CF-based architecture offers significantly higher coverage throughout the SINR threshold range. This highlights the potential of CF-mMIMO-enabled LEO SatCom systems to fully exploit dense satellite deployments, thereby ensuring robust and enhanced coverage.
}

\textcolor{black}{
Fig. \ref{fig:Fig_BF} presents a comparison of the coverage probability under Nakagami-$m$ fading channels with beamforming (B.F.) and without beamforming (w/o B.F.). Both simulation and analytical results are provided for $m=1$ and $m=4$. It is evident from the figure that, for a given SINR threshold, beamforming consistently enhances coverage performance across both fading scenarios. For instance, when $m=1$, the curve with beamforming significantly outperforms its non-beamforming counterpart, with a noticeable coverage gap of over $0.2$ around $\gamma_\text{th}=3$ dB. A similar trend is observed for $m=4$, although the performance gap is less pronounced due to the reduced severity of the fading. Moreover, when comparing different fading parameters, the improvement offered by beamforming is more substantial under severe fading conditions, i.e., smaller $m$. This suggests that beamforming is particularly effective in mitigating deep fades, thereby improving the reliability of the CF-mMIMO SatCom network.
}

\captionsetup{font={scriptsize}}
\begin{figure}[tp]
\begin{center}
\setlength{\abovecaptionskip}{+0.2cm}
\setlength{\belowcaptionskip}{-0.4cm}
\centering
  \includegraphics[width=3.4in, height=2.7in]{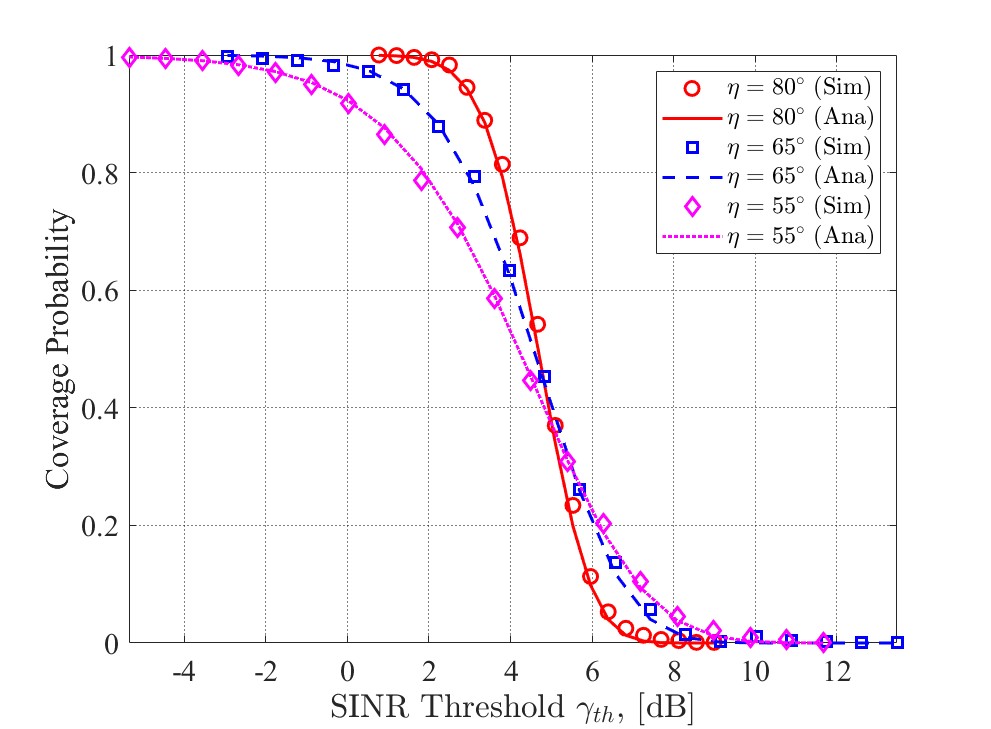}
\renewcommand\figurename{Fig.}
\caption{\scriptsize Coverage probability versus SINR threshold and $\eta$.}
\label{fig:Fig_eta}
\end{center}
\end{figure}

\captionsetup{font={scriptsize}}
\begin{figure}[tp]
\begin{center}
\setlength{\abovecaptionskip}{+0.2cm}
\setlength{\belowcaptionskip}{-0.4cm}
\centering
  \includegraphics[width=3.4in, height=2.7in]{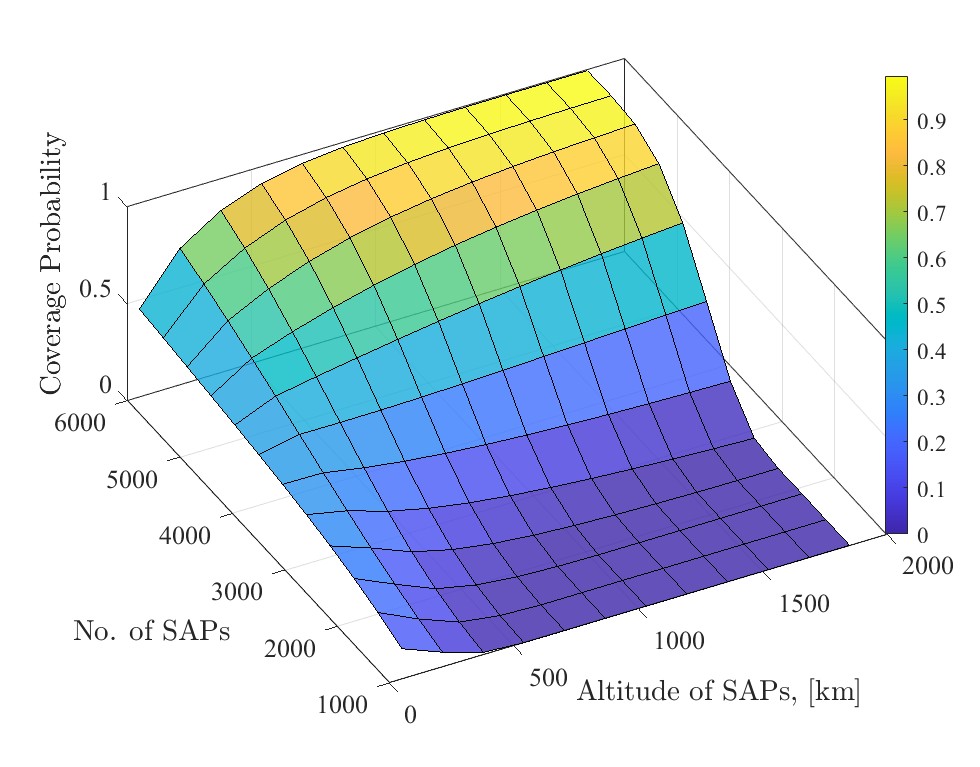}
\renewcommand\figurename{Fig.}
\caption{\scriptsize Coverage probability under different altitudes and numbers of SAPs in a 3D plot, where $\gamma_\text{th}=3$ dB.}
\label{fig:Fig_nbrSAP_altitude}
\end{center}
\end{figure}

\captionsetup{font={scriptsize}}
\begin{figure}[tp]
\begin{center}
\setlength{\abovecaptionskip}{+0.2cm}
\setlength{\belowcaptionskip}{-0.4cm}
\centering
  \includegraphics[width=3.4in, height=2.7in]{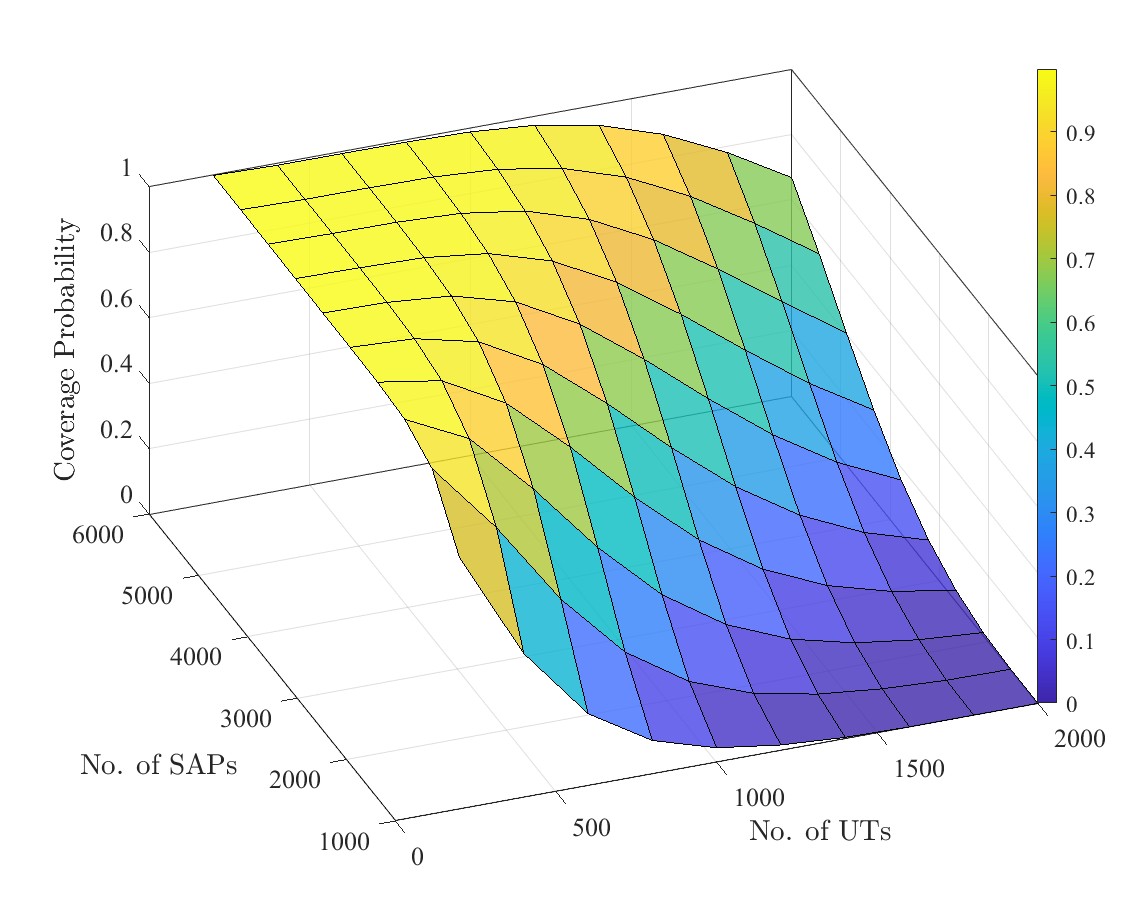}
\renewcommand\figurename{Fig.}
\caption{\scriptsize Coverage probability under different numbers of SAPs and numbers of UTs in a 3D plot, where $\gamma_\text{th}=3$ dB.}
\label{fig:Fig_nbrUT_nbrSAP}
\end{center}
\end{figure}

Then, we examine the impact of the service range, which is determined by the dome angle $\eta$ on the coverage probability. Fig. \ref{fig:Fig_eta} presents the coverage probability for different values of $\eta$, i.e., $\eta=55^{\circ},65^{\circ},80^{\circ}$ \cite{pachler2021updated}. 
It can be observed that at lower SINR thresholds, e.g., below $5$ dB, a larger $\eta$ enhances the coverage probability, as the expanded service range increases the number of available service satellites. This improves received signal strength and mitigates deep fading through cooperative transmission, even in the presence of moderate MUI.
\textcolor{black}{
However, as the SINR threshold increases from $5$ dB to $12$ dB, the curves cross over and the trend reverses. For smaller $\eta$, i.e., $55^\circ$, the limited service range reduces both the number of serving signals and the MUI. 
In this regime, the serving signals dominate due to stronger direct links, while ISI remains relatively moderate because of longer propagation distances and attenuation by side-lobe gain. 
Conversely, for larger $\eta$, i.e., $80^\circ$, although more service satellites are available and enhance the desired signal at low thresholds, the accompanying MUI grows rapidly because it is also subject to the main-lobe gain. Meanwhile, ISI does not scale at the same rate under side-lobe gain. 
This imbalance leads to an earlier degradation of coverage when the SINR threshold becomes more stringent, which explains the crossover between different $\eta$ curves. Consequently, the benefits of expanding the service dome angle become saturated and can even reverse in the high-SINR regime.}
Eventually, for all values of $\eta$, the coverage converges to near zero as the threshold increases, as the strict SINR requirement becomes increasingly difficult to meet. It is therefore suggested that a proper SINR threshold be carefully chosen to ensure the desired per-user and system performance.

\textcolor{black}{
Next, the dual influences of the orbital altitude and the number of SAPs are shown in Fig. \ref{fig:Fig_nbrSAP_altitude} in a three-dimensional (3D) view. For a fixed altitude within the given range, increasing the number of SAPs always brings higher coverage.
However, when the number of SAPs is less than approximately $3500$, coverage decreases with rising altitude, while it increases with altitude when the number of SAPs exceeds $3500$.
Given $\gamma_\text{th}=3$ dB, when the total number of SAPs is relatively small, the increase in the number of SAPs cannot offset the desired signal reduction due to the extended propagation distances from increased altitude.
In addition, we study the effect of the number of UTs and the number of SAPs on the coverage probability in Fig. \ref{fig:Fig_nbrUT_nbrSAP}. It is shown that higher coverage is achieved with fewer UTs and more SAPs. This is consistent with terrestrial CF-mMIMO systems, where the performance of an individual UT is enhanced when a large number of APs serve a smaller number of UTs.
}

\captionsetup{font={scriptsize}}
\begin{figure}[tp]
\begin{center}
\setlength{\abovecaptionskip}{+0.2cm}
\setlength{\belowcaptionskip}{-0.0cm}
\centering
  \includegraphics[width=3.4in, height=2.7in]{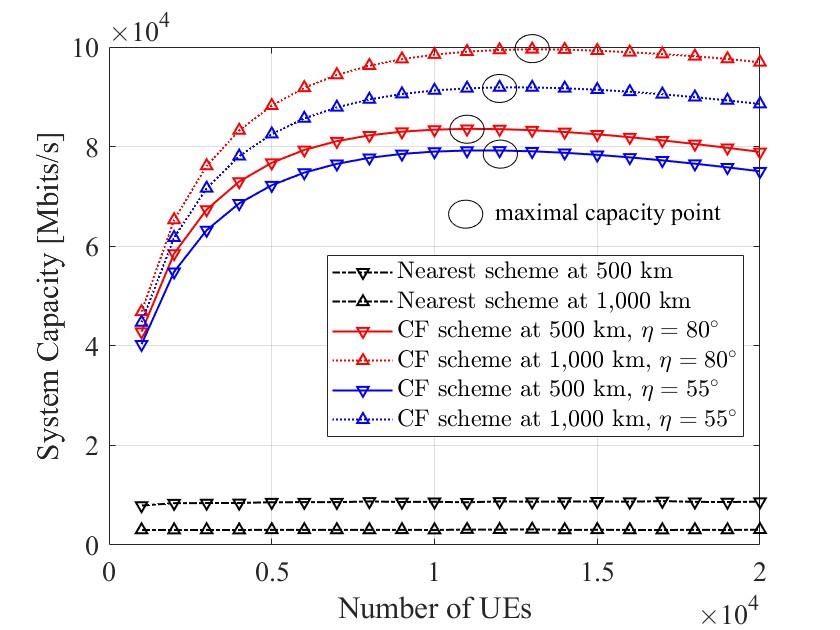}
\renewcommand\figurename{Fig.}
\caption{\scriptsize Ergodic system capacity versus number of UTs, where $B=30$ MHz.}
\label{fig:Fig_syst_cap}
\end{center}
\end{figure}

\subsection{System and Individual Capacity}

Fig. \ref{fig:Fig_syst_cap} compares the ergodic system capacity of the CF-mMIMO SatCom network, i.e., CF scheme, with that of the nearest-satellite-service scheme in \cite{park2022tractable},  where a UT is served only by the nearest satellite.
\textcolor{black}{
There exists a relationship between coverage probability and spectral efficiency according to a mathematical basis \cite[Appendix C]{jo2012heterogeneous}
\begin{equation}
    \mathbb{E}\{X\}=\int_{0}^{\infty}\mathbb{P}\{X>x\}dx,
\end{equation}
for $X>0$, which will bring 
\begin{equation}
    \begin{aligned}
        \mathbb{E}\left\{ \mathrm{log}_2 (1+\mathrm{SINR}_k) \right\}
        &= \int_{0}^{\infty} \mathbb{P} \left\{ \mathrm{SINR}_k > 2^t-1 \right\} dt,
    \end{aligned}
\end{equation}
for the $k$-th UT \cite{li2024analytical}.
It should also be noted that potential channel estimation overhead and downlink transmission durations will be considered.
}
Then, the system capacity of the CF scheme is defined as
\begin{equation}
    \begin{aligned}
        \text{C}_\text{system}^\text{CF}
        &= N_\text{U} B \frac{\tau_\text{c}-\tau_\text{p}}{\tau_\text{c}} \mathbb{E} \left\{\rm{log}_2 (1+\text{SINR}_\textit{k}) \right\} \\
        &= N_\text{U} B \frac{\tau_\text{c}-\tau_\text{p}}{\tau_\text{c}}  \int_{0}^{\infty} \mathbb{P}_{k}^{\text{cov}} \left\{ \text{SINR}_k > 2^t -1 \right\} dt,
    \end{aligned}
    \label{Formula:syst_cap_cf}
\end{equation}
and that of the nearest-satellite-service scheme is given by
\begin{equation}
    \begin{aligned}
        \text{C}_\text{system}^\text{Nearest}
        &= N_\text{U} B \int_{0}^{\infty} \mathbb{P}_{k}^{\text{cov}} \left\{ \text{SINR}_k > 2^t -1 \right\} dt,
    \end{aligned}
    \label{Formula:syst_cap_Nearest}
\end{equation}
where $N_\text{U}$ denotes the total number of UTs on the earth's surface, $\tau_\text{p}=200$, and $\tau_\text{c}=500$.
\textcolor{black}{
In the nearest-satellite-service scheme, for altitudes of $500$ km and $1,000$ km, a slight increase in the system capacity can be observed when the number of UTs grows from $1,000$ to $3,000$ before reaching a plateau.
The system capacity at $500$ km is greater than that at $1,000$ km.
In contrast, for the CF scheme, the system capacity at $1,000$ km is greater than that at $500$ km for both $\eta=80^{\circ}$ and $\eta=55^{\circ}$. This is because a higher altitude at a fixed $\eta$ incorporates more SAPs into the network. 
}

\textcolor{black}{
Moreover, there is an optimal number of UTs for each $\eta$ at a certain altitude. 
For instance, for $\eta=80^{\circ}$ and an altitude $1,000$ km, when the total number of UTs increases to more than $13,000$, the system capacity of the CF scheme begins to decline, though it still significantly outperforms the nearest-satellite-service scheme.
It is worth mentioning that the peak in system observed in Fig. \ref{fig:Fig_syst_cap} results from the inherent trade-off between user multiplexing and MUI. Initially, increasing the number of UTs enhances overall throughput, as the system can support more simultaneous transmissions. However, beyond a certain point, the resulting increase in MUI significantly degrades SINR at the UTs, leading to a decrease in total capacity. This trend highlights the importance of optimal user scheduling and interference-aware resource management in CF-mMIMO LEO SatCom networks.
}

The above schemes are also compared regarding per-user capacity in Fig. \ref{fig:Fig_perUT_cap}, which is defined as
\begin{equation}
    \text{C}_\text{per-user}^\text{CF} = \text{C}_\text{system}^\text{CF} / N_\text{U},
    \label{Formula:per_cap_cf}
\end{equation}
\begin{equation}
    \text{C}_\text{per-user}^\text{Nearest} = \text{C}_\text{system}^\text{Nearest} / N_\text{U},
    \label{Formula:per_cap_Nearest}
\end{equation}
for the CF scheme \cite{ngo2017cell} and nearest-satellite-service scheme, respectively.
\textcolor{black}{
All cases exhibit a decreasing trend. In particular, the CF scheme with $\eta=80^{\circ}$ and an altitude $1,000$ km yields the highest per-user capacity across all numbers of UTs. This can also be attributed to the reception of more desired signals from a larger number of SAPs when the altitude is higher and the dome angle $\eta$ is larger. All CF schemes provide higher per-user capacity compared to the nearest-satellite-service scheme.
However, the per-user capacity tends to approach that of the nearest-satellite-service scheme. This suggests that the significant performance gain of CF-mMIMO SatCom may diminish as the number of UTs increases and under heavy user loads.
}

\captionsetup{font={scriptsize}}
\begin{figure}[tp]
\begin{center}
\setlength{\abovecaptionskip}{+0.2cm}
\setlength{\belowcaptionskip}{-0.0cm}
\centering
  \includegraphics[width=3.4in, height=2.7in]{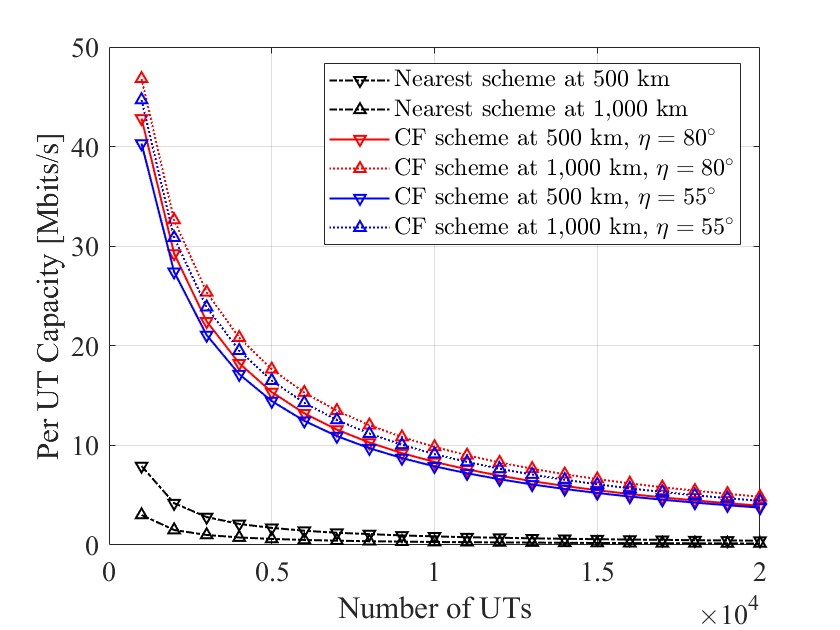}
\renewcommand\figurename{Fig.}
\caption{\scriptsize Ergodic per-user capacity versus number of UTs, where $B=30$ MHz.}
\label{fig:Fig_perUT_cap}
\end{center}
\end{figure}

\subsection{Future Explorations Discussion}
This paper focuses on the fundamental downlink performance of CF-mMIMO in LEO satellite mega-constellations. Several promising directions remain open for exploration, including coordinated SAP clustering, multi-antenna extensions, and advanced beamforming designs.

\subsubsection{Coordinated SAPs Clustering}
In future satellite networks, especially dense deployment scenarios, the fully distributed design in this paper may be insufficient when suppressing MUI. To address such potential limitations, a lightweight coordinated clustering and limited CSI exchange framework may be considered for CF-mMIMO LEO SatCom networks to improve scalability \cite{bjornson2020scalable}.

Geographically or topologically adjacent SAPs could dynamically form coordination clusters based on spatial correlation, user density, or overlap in signal coverage. Within each cluster, SAPs can exchange partial CSI, such as large-scale fading statistics or quantized channel direction information, rather than full instantaneous CSI. This limited information exchange enables more informed beamforming and scheduling strategies \cite{ammar2021user}, thus improving interference suppression with minimal coordination overhead. 
Given the constraints of LEO satellite systems, such as restricted inter-satellite communication and processing delays, such a lightweight scheme could strike a practical balance between performance and scalability.

\subsubsection{Multi-antenna Extensions}
Both SAPs and UTs can be extended to support multiple-antenna configurations. Herein, the assumption of a single antenna is made to simplify the initial analysis and provide clarity in stochastic geometry modeling. To partially reflect the directional enhancement brought by multi-antenna beamforming, this paper incorporates an antenna gain pattern that distinguishes between main-lobe and side-lobe gains. 
\textcolor{black}{
Such main-lobe/side-lobe modeling has been adopted in stochastic geometry and satellite communications analyses \cite{park2022tractable,jia2022uplink,jia2023analytic}, which can capture key beamforming benefits such as improved signal strength and reduced interference. 
This enables a clear examination of basic system performance without adding complexity of multi-antenna settings.
}

While the single-antenna setup may not fully reflect all features of a conventional massive MIMO system, it provides a practical and scalable foundation, especially relevant for scenarios with limited satellite payload or simplified deployment requirements. The proposed model can be extended in future studies to incorporate multiple antennas per SAP, aligned with the conventional massive MIMO configurations.

\subsubsection{Beamforming Designs} 
Conjugate beamforming is a well-established technique to simplify the distributed beamforming process with relatively low implementation complexity and minimal backhaul signaling requirements \cite{yang2013performance}.
While alternative linear precoding schemes such as zero-forcing (ZF) beamforming \cite{zhang2019cell}, MMSE beamforming \cite{palhares2021robust}, and beamforming with feedback channels \cite{dakkak2022evaluation} can offer improved interference suppression, they typically require significantly higher coordination overhead and centralized processing\footnote{Note that the implementation complexity is especially relevant in the context of LEO satellite mega-constellations with stringent on-board processing constraints and rapid channel variation. Thus, the trade-off between the complexity of implementation and the achievable performance of these techniques is of interest and remains an important direction for future research.} compared to conjugate beamforming\footnote{Our current analysis based on conjugate beamforming provides a useful baseline for more sophisticated precoding techniques and for understanding the fundamental performance characteristics of CF-mMIMO in LEO satellite networks. The adoption of more sophisticated precoding techniques is left for further work in non-terrestrial CF-mMIMO.} \cite{ngo2017cell}.

Generally, our antenna and beamforming settings serve as a reasonable starting point for stochastic geometry-based CF-mMIMO for LEO satellite mega-constellations. These settings establish a benchmark and examine fundamental system behaviors for future theoretical analysis and performance comparisons.

\section{Conclusion}
In this paper, we analyzed the CF-mMIMO LEO SatCom network in terms of coverage and capacity from a system-level perspective. 
Using tools from stochastic geometry, we modeled the satellite network in two SPPPs and derived the coverage probability in fading scenarios, accounting for key network parameters, including the Nakagami fading parameter, path-loss factor, orbital altitude, number of SAPs, and service range influenced by the dome angle. The Laplace transform and relevant numerical approximation methods were used to analyze the desired and interference signals. Operation mechanisms regarding technical explanations for potential considerations were also discussed to facilitate practical implementations.

Although concise-form solutions are not always feasible in the context of CF-mMIMO LEO SatCom networks, the derived expressions facilitate system evaluations that reveal performance trends and key trade-offs, maintaining the analytical value of the stochastic framework.
Based on these numerical and analytical results, we find that beamforming plays an essential role in network performance, and the direct distance-related propagation path dominates signal quality. Increasing the service range, orbital altitude, and number of SAPs improves coverage performance.
Furthermore, while there is an optimal number of UTs to maximize the system capacity, the capacity for each individual UT declines as the number of UTs increases. Although additional SAPs can be accommodated by raising their orbital altitudes, the corresponding increase in costs and communication latency should be considered in real-world satellite network deployment.

\bibliographystyle{IEEEtran}
\bibliography{references.bib}

\end{document}